\documentclass[journal,twocolumn]{IEEEtran}
\usepackage[final]{graphicx}

\usepackage{amsmath,epsfig,amssymb,verbatim,amsopn,cite,subfigure,multirow,amsthm}
\usepackage{balance}
\usepackage[usenames,dvipsnames]{color}
\usepackage[all]{xy}  
\usepackage{url}
\usepackage{amsfonts}
\usepackage{amssymb}
\usepackage{epsfig}
\usepackage{slashbox}
\usepackage{bm}
\usepackage{epsfig}
\usepackage{amsmath,amssymb, amstext}
\usepackage{graphicx}
\usepackage{epstopdf}
\usepackage{algorithm}
\usepackage{algorithmic}
\usepackage{tabularx,booktabs}
\usepackage{url}
\usepackage{multirow}
\usepackage{array}
\usepackage{makecell}
\usepackage{balance}

\newtheorem{definition}{Definition}

\newtheorem{lemma}{Lemma}

\IEEEaftertitletext{\vspace{-1.5\baselineskip}}
\begin{document}

\title{D2D-aided LoRaWAN LR-FHSS in Direct-to-Satellite IoT Networks 
\thanks{Alireza Maleki, Ha H. Nguyen, and Ebrahim Bedeer are with the Department of Electrical and Computer Engineering, University of Saskatchewan, Saskatoon, Canada S7N5A9. Emails: \{alireza.maleki, ha.nguyen, and e.bedeer\}@usask.ca.}
\thanks{R. Barton is with Cisco Systems Inc. Emails: robbarto@cisco.com.}
\thanks{This work was supported by NSERC/Cisco Industrial Research Chair program.}
\thanks{The co-authors dedicate this paper to Prof. Ha H. Nguyen, who passed away before the submission of the paper.}
}
\author{Alireza Maleki, Ha H. Nguyen, Ebrahim Bedeer, and Robert Barton}

\maketitle

\begin{abstract}
In this paper, we present a device-to-device (D2D) transmission scheme for aiding long-range frequency hopping spread spectrum (LR-FHSS) LoRaWAN protocol with application in direct-to-satellite IoT networks. We consider a practical ground-to-satellite fading model, i.e. shadowed-Rice channel, and derive the outage performance of the LR-FHSS network. With the help of network coding, D2D-aided LR-FHSS transmission scheme is proposed to improve the network capacity for which a closed-form outage probability expression is also derived. The obtained analytical expressions for both LR-FHSS and D2D-aided LR-FHSS outage probabilities are validated by computer simulations for different parts of the analysis capturing the effects of noise, fading, unslotted ALOHA-based time scheduling, the receiver's capture effect, IoT device distributions, and distance from node to satellite. The total outage probability for the D2D-aided LR-FHSS shows a considerable increase of $249.9\%$ and $150.1\%$ in network capacity at a typical outage of $10^{-2}$ for DR6 and DR5, respectively, when compared to LR-FHSS. This is obtained at the cost of minimum of one and maximum of two additional transmissions per each IoT end device imposed by the D2D scheme in each time-slot.
\end{abstract}

\section{Introduction}

\IEEEPARstart{T}{he} Internet of Things (IoT) is a sizable application domain that has recently been formed thanks to advancements in automation and miniaturization. Among different types of IoT applications, massive machine-type communication (mMTC), mostly refers to data-collecting applications where a large number of low-power endpoints (such as inexpensive sensors), are used to continually and sparingly (usually sporadically) transfer small amounts of data to the cloud or fusion centers \cite{Vaezi_2022}. Various low-power wide-area network (LPWAN) solutions, such as LoRaWAN \cite{LoRa_White}, SigFox \cite{SigFox}, and the third-generation partnership project (3GPP)-approved narrowband IoT (NB-IoT) \cite{NBIoT}, have been suggested as candidate solutions to meet such needs. For example, in LPWANs, the supported link range is often in the order of few kilometers and the low cost of the LPWAN network infrastructure generally comes at the trade-off of extremely low bit rates, message rates, and frame payload sizes \cite{Malik_2022,Stusek_2022}.

In this paper, we focus on LoRaWAN which is rapidly gaining industrial popularity throughout the globe. The LoRa Alliance\footnote{https://lora-alliance.org/} established LoRaWAN, the network protocol, where LoRa refers to the chirp spread spectrum (CSS) technology utilized at the physical layer (patented by Semtech). Despite being an open standard and easy to implement, the LoRaWAN protocol has its limits, particularly in dense network installations where the mandated duty cycle and the ALOHA-based medium access control (MAC) protocol severely restrict the overall network capacity. On the other hand, due to an increasing demand for the deployment of numerous IoT devices, also known as end devices (EDs), in various locations on Earth, traditional terrestrial LoRa networks may become extremely dense in some cases (such as offshore wind farm monitoring using sensors installed on the wind turbines) \cite{Ullah_2021}. 

The above-mentioned issues give rise to the introduction of direct-to-satellite IoT (DtS-IoT) communication \cite{Fraire_2019,Palattella_2018}. Particularly, DtS-IoT lowers the reliance on ground gateways (GWs) and, as a result, makes connectivity to remote locations simpler. Technologies such as LoRaWAN and NB-IoT have been investigated to demonstrate their feasibility of being exploited for low-power DtS-IoT access mode \cite{Doroshkin_2019,Ouvry_2018,Deng_2017}. New start-up businesses like Wyld\footnote{https://www.wyldnetworks.com/} and Lacuna Space\footnote{https://lacuna.space/} are working on deploying specific satellite constellations for IoT and have already shown how a LoRaWAN GW can be integrated on a low earth orbit (LEO) satellite to enable worldwide IoT connection \cite{Alvarez_2022,Afhamisis_2022}. Additionally, compared to traditional satellite networks, the use of nano-satellites along with LPWAN technologies over a satellite link offers more affordable and delay-tolerant IoT connectivity solutions \cite{Almonacid_2017,Akyildiz_2019}. More recently, to solve the extremely long-range and large-scale communication scenarios introduced by DtS-IoT access mode, a new PHY layer transmission called long-range frequency hopping spread spectrum (LR-FHSS) \cite{9,10} is developed, which is now on the agenda for upcoming space IoT initiatives \cite{Alvarez_2022}. The core of LR-FHSS is a fast frequency hopping method that offers the same radio link budget as LoRa while enabling increased network capacity.  

The use of LR-FHSS is specifically demonstrated in \cite{10} to support total loads at peak efficiency of $3.5$ million and $1.48$ million packets/hour for two data rates specified in the Europe frequency band, namely DR8 and DR9, respectively. These results reflect around $36$-time and $15$-time capacity increases compared to the greatest data rate (DR0) of $96000$ packets/hour offered by conventional LoRa. It is worth mentioning that the authors in \cite{10} assumed ideal channel conditions with no fading and noise presence. The authors in \cite{12} develop a mathematical model for packet delivery probability that only considers the frequency hopping algorithm of the LR-FHSS scheme and ignores the effects of fading and path-loss, which are unreliable in real-world applications. Without offering any theoretical derivations, the simulator model just takes into account the capture effect (fading and path loss included). Also noteworthy is the fact that this approach fully disregards the noise effect. To overcome these shortcomings, the authors in \cite{Self_2022} presented a mathematical analysis framework taking into account the path-loss, channel fading, and noise and verified the obtained closed-form expression for outage probability of LR-FHSS by simulation results. The results show a significant improvement in terms of network capacity compared to LoRa networks in a dense DtS IoT scenario. However, the channel fading used in \cite{Self_2022} is modeled by the Nakagami-m distribution which only depicts the network performance in ``good'' ground-to-satellite link state \cite{15}. Therefore, an analysis with a complete ground-to-satellite channel modeling is still lacking. A scalable and energy-efficient DtS-IoT is investigated in \cite{Alvarez_2022} where the LR-FHSS physical layer is included in the study. It specifically suggests uplink transmission policies that make use of satellite trajectory data. The theoretical mixed integer linear programming (MILP) model that frames these schemes serves as both an upper constraint on performance and a tool for developing scheduled DtS-IoT solutions. According to the provided results, which are solely based on a simulation tool developed in Python, trajectory-based policies can improve the network capacity and in certain cases, can increase scalability without costing extra energy. Most recently, the authors in \cite{Ullah_2022} support the potential use of mMTC and LPWAN technologies for DtS communications in Arctic autonomous ships applications. Simulations are performed based on actual ship and satellite positions, traffic patterns, and the LoRaWAN connectivity model to test this theory and examine the possible performance and impact of various design and configuration options. The results show the viability of the proposed strategy and indicate how various parameters affect the connectivity performance for both the traditional LoRa and the LR-FHSS schemes. Notably, the packet delivery is significantly enhanced by combining multiple LoRaWAN LPWAN connections with multiple satellite visibility. Considering these works, one can see that a complete analytical model taking into account a practical ground-to-satellite channel fading for the use of LR-FHSS in DtS IoT network is still missing in the literature. In the first part of this paper, we develop such an analytical model capturing all practical aspects of DtS IoT link.

To keep up with the rapid growth of IoT network implementations in today's world with an impressive rate of $12\%$ annually (estimated $125$ billion connected IoT EDs in 2030 \cite{IoTNum}), the need for more sizable network coverages is inevitable. On the other hand, low deployment costs have been a major factor in the development of massive IoT sensor networks \cite{Potdar_2009}. In terms of network design, this poses a dilemma for the densification of IoT gateways: while a low density of gateways can significantly lower deployment costs, it makes it nearly impossible for all IoT nodes to have dependable access to the network \cite{Vaezi_2022}. Although in LPWANs, IoT devices are directly connected to a radio GW via a simple star network topology, for low-cost wireless IoT networks, this has primarily motivated wireless engineers to create mesh-type network topologies to facilitate exploiting device-to-device (D2D) data communication schemes. In essence, each communication node in the network participates in both the communication of its own data and the support of the transfer of data from other nodes.

To cope with the mentioned challenges, D2D relaying has been researched and specified for many different wireless technologies for a long time to extend coverage, improve the battery life of sensors, and enable out-of-coverage communication between devices. The current LoRaWAN specification does not feature a mechanism to enable D2D communication among IoT EDs \cite{Mikhaylov_2017}; however, one may expect that D2D will be included in the specification in the near future to allow EDs to transmit their data cooperatively and benefit from the spatial diversity provided by emulating the impacts of multiple antennas in a network composed of single antenna devices. The authors in \cite{Mikhaylov_2017} encourage the use of D2D communications in a LoRaWAN network. A network-aided D2D communication protocol is proposed and its viability is demonstrated by layering it on top of a commercial transceiver that has received LoRaWAN certification. From the results of \cite{ Mikhaylov_2017}, it can be seen that when compared to traditional LoRaWAN data transmission techniques, D2D communications can improve the performance in terms of the time and energy. It can be interpreted that applications requiring high coverage can leverage the planned LoRaWAN D2D communications, such as use cases in distributed smart grid deployments for management and trade. Moreover, in a more complex cooperative mechanism, the EDs may use the idea of network coding \cite{Rebelatto_2012,Xiao_2010,Wu_2015,Oliveira_2022} and broadcast linear combinations, carried out across a finite field ${\rm{GF}}(q)$, of several frames rather than merely serving as routers by relaying one frame at a time, aiming to increase reliability rather than throughput. In \cite{ Oliveira_2022}, a network-coded cooperation LoRa scheme is proposed. In case of EDs being able to communicate with one another using D2D technology, the outage probability and energy efficiency (EE) of a LoRa network are assessed. When modeling the outage likelihood of an ED, both connection and collision probabilities are taken into account. Also, an analysis of the EE using a realistic power consumption model is presented. The results show that the proposed approach can lead to a considerable improvement in terms of both outage probability and EE when compared to a conventional LoRa network. Given this, we propose a transmission scheme based on the integration of D2D communication with the LR-FHSS approach utilizing a network-coding scheme and obtain the outage performance results for a practical DtS-IoT scenario. 

The contributions of this paper can be summed up as follows:

\begin{itemize}
\item We present a detailed analytical procedure to obtain a closed-form expression for the outage probability of LR-FHSS scheme differently from \cite{10,12} by taking into account the effects of important parameters such as practical fading model (shadowed-Rician), path loss, and noise.
\item We present a D2D network coding-based LR-FHSS transmission scheme in which the EDs communicate with each other as well as the IoT gateway installed on a LEO satellite using LoRa and LR-FHSS, respectively. We also provide a closed-form expression of the outage probability of the D2D-aided LR-FHSS.
\item We indicate analytically and through simulations that the D2D-aided LR-FHSS can provide up to $249.9\%$ and $150.1\%$ in network capacity at a typical outage of $10^{-2}$ for DR6 and DR5, respectively, when compared to the LR-FHSS scheme which is obtained at the cost of additional transmissions by each IoT ED.  
\end{itemize}
 
The rest of the paper is organized as follows. In Section II, the system model is presented. The Outage performance of the LR-FHSS scheme is derived in Section III. Section IV includes the D2D-aided LR-FHSS scheme and its outage probability analysis. Numerical results are provided in Section V. Finally, in Section VI, the paper is concluded.

\begin{table*}[t!]
	\caption{Parameters for LR-FHSS in the United States region \cite{9}.}
	\label{tab_LR_FHSS}
	\centering
	\begin{tabularx}{\textwidth}{Xcc}
		\toprule[1.0pt]
		LoRaWAN data rate alias & DR5 & DR6 \\
		LR-FHSS number of channels & $8$ & $8$ \\
		LR-FHSS operating channel width (OCW), in kHz & $1523$ & $1523$ \\
		LR-FHSS occupied bandwidth (OBW), in Hz & $488$ & $488$ \\
		Minimum separation between LR-FHSS hopping carriers, in kHz	& $25.4$ & $25.4$ \\
		Number of physical carriers available for\ frequency hopping in each OCW channel	 & $3120$ ($52\times60$)	& $3120$ ($52\times60$) \\
		Number of physical carriers usable for frequency hopping per end-device transmission &	 $60$	& $60$ \\
		Coding rate	& $1/3$	& $2/3$ \\
		Physical bit rate (bits/s)	& $162$ &	$325$ \\
		Maximum MAC payload size (bytes)	& $125$ & 	$125$ \\
		Maximum MAC payload fragments	& $130$ &	$65$ \\
		Header replicas	& $3$	& $2$\\
		Header duration per replica (seconds)	& $0.233$ &	 $0.233$\\
		Time on air (seconds)	& $0.70+13.26$	& $0.47+6.63$ \\
		\bottomrule[1.0pt]
	\end{tabularx}
\end{table*}

\section{System model}
Without loss of generality, we take into account  the $902–928$ MHz LR-FHSS specification for the United States region as described in Table \ref{tab_LR_FHSS}. With center frequencies of $903+1.6n$ MHz ($n=0,1,\dots,7$) and $1.523$ MHz channel bandwidth, the frequency range of $902-928$ MHz comprises $8$ LR-FHSS channels, i.e., operating channel width (OCW). Additionally, each OCW has several sub-carriers with occupied bandwidths (OBW) of $488$ Hz that are used for frequency hopping purposes. Hence, the number of sub-carriers available for creating a hopping sequence inside a single OCW is $(1.523\times 10^6)/488=3120$. The minimum frequency spacing between any two consecutive hops, however, must be $25.4$ kHz. Essentially, this requirement results in a hopping grid with $G=52$ groups and $S=60$ sub-carriers for each group to be exploited in a hopping sequence. Specifically, when an ED wants to send its packet, first, a group number is chosen randomly from all of the $52$ groups available. Then, a hash function-based random hopping sequence is generated using the frequencies inside the selected group with a length equal to the sum of header replicas and payload plus cyclic redundancy check (CRC) \cite{9} fragments. Note that hereinafter, we use the term ``payload'' for ``payload plus CRC'' for ease of notation. Using the data from the header, the gateway at the receiving end recognizes and reassembles the packet payload \cite{9,10}. Fig. \ref{Freq_Time} shows the frequency-time occupation of a LR-FHSS modulated packet inside a single OCW for a $90$ byte payload in DR5 transmission mode. As can be seen, in the DR5 specification, we have $3$ header replicas and the payload part is fragmented into $90/6=15$ segments as specified by \cite{9}.

\begin{figure}[t!]
  \centering
  \includegraphics[width=\linewidth]{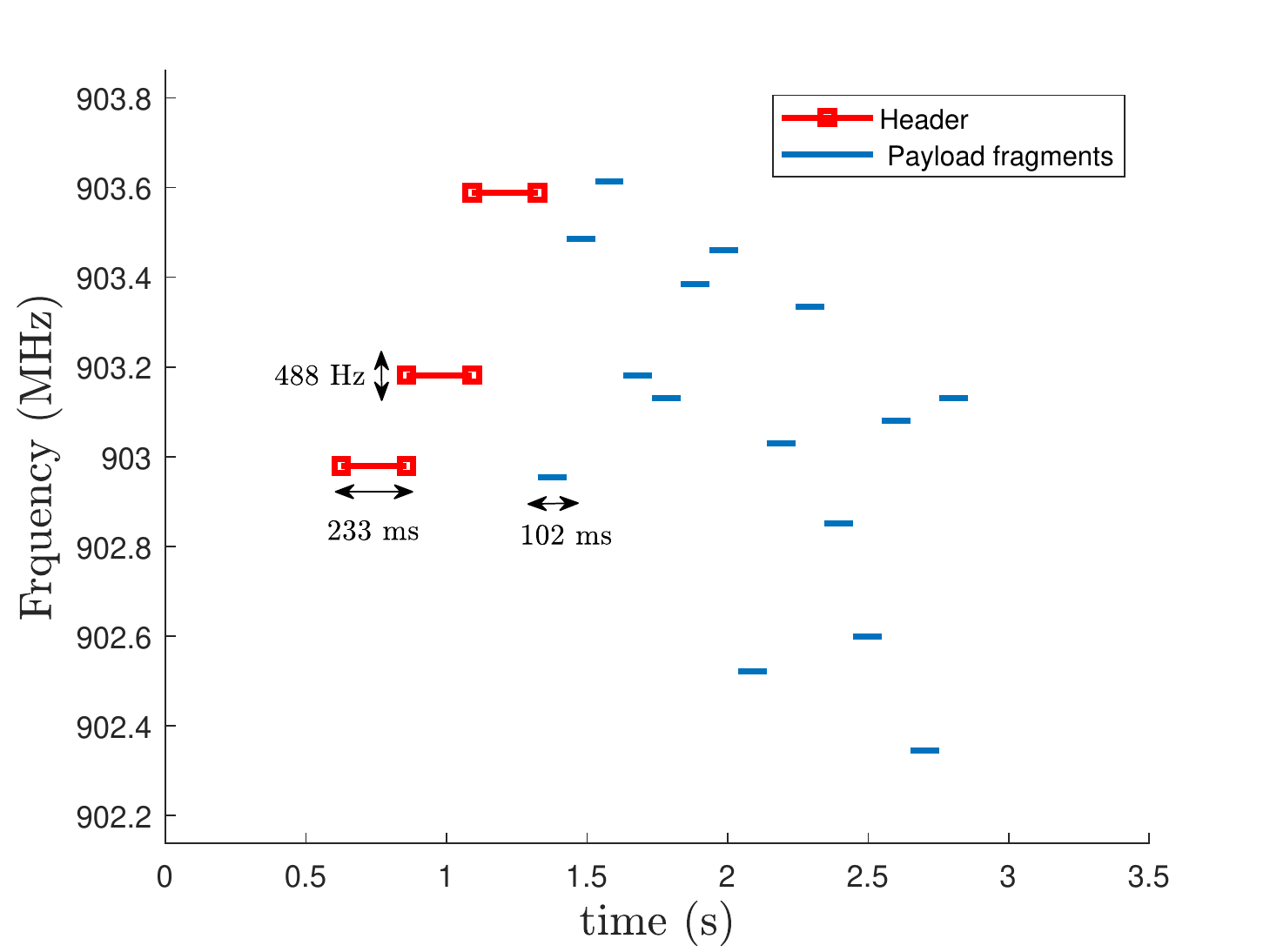}
  \caption{Frequency-time demonstration of a LR-FHSS packet.}
  \label{Freq_Time}
\end{figure}

The network diagram for the proposed D2D-aided LR-FHSS system is shown in Fig. \ref{LoRa_FHSS}. As can be seen, there are three main links in the network, i.e., D2D link, ED to satellite link, and satellite to network server (NS) link. In this paper, we focus only on D2D and ED to satellite link and the characteristics of the backhaul link is out of scope of this work. Assume that a LEO satellite is orbiting at the height of $H_{\rm s}$ with a footprint in form of a circular region with the radius of $\mathcal{R}_{\rm s}$. Moreover, the EDs are assumed to be Class A devices following the uplink LoRaWAN unslotted ALOHA-based grant-free channel access scheme \cite{SX1272} with a maximum duty cycle of $1\%$. Therefore, we can formulate the duty cycle for a given LR-FHSS data rate (DR) as:

\begin{equation}
\label{duty_cycle}
\Delta({\rm{DR}}) = N_{\rm t} \frac{{\rm{ToA}_p}}{T},
\end{equation}
where $N_{\rm t}$ and $T$ denote number of transmissions in a time-slot and duration of each time-slot, respectively. Moreover, the ${\rm{ToA_p}}$ represents the time-on air of the LR-FHSS packet and can be formulate as:

\begin{equation}
\label{ToA}
{\rm{ToA_p}}=N_{\rm{HDR}}T_{\rm{HDR}}+N_{\rm{PL}}T_{\rm{PL}},
\end{equation}
where $N_{\rm{HDR}}$, $T_{\rm{HDR}}$, $N_{\rm{PL}}$, and $T_{\rm{PL}}$ are the number of header replicas, the header duration, the number of payload fragments, and the duration of a payload fragment, respectively.

\begin{figure}[t!]
  \centering
  \includegraphics[width=\linewidth]{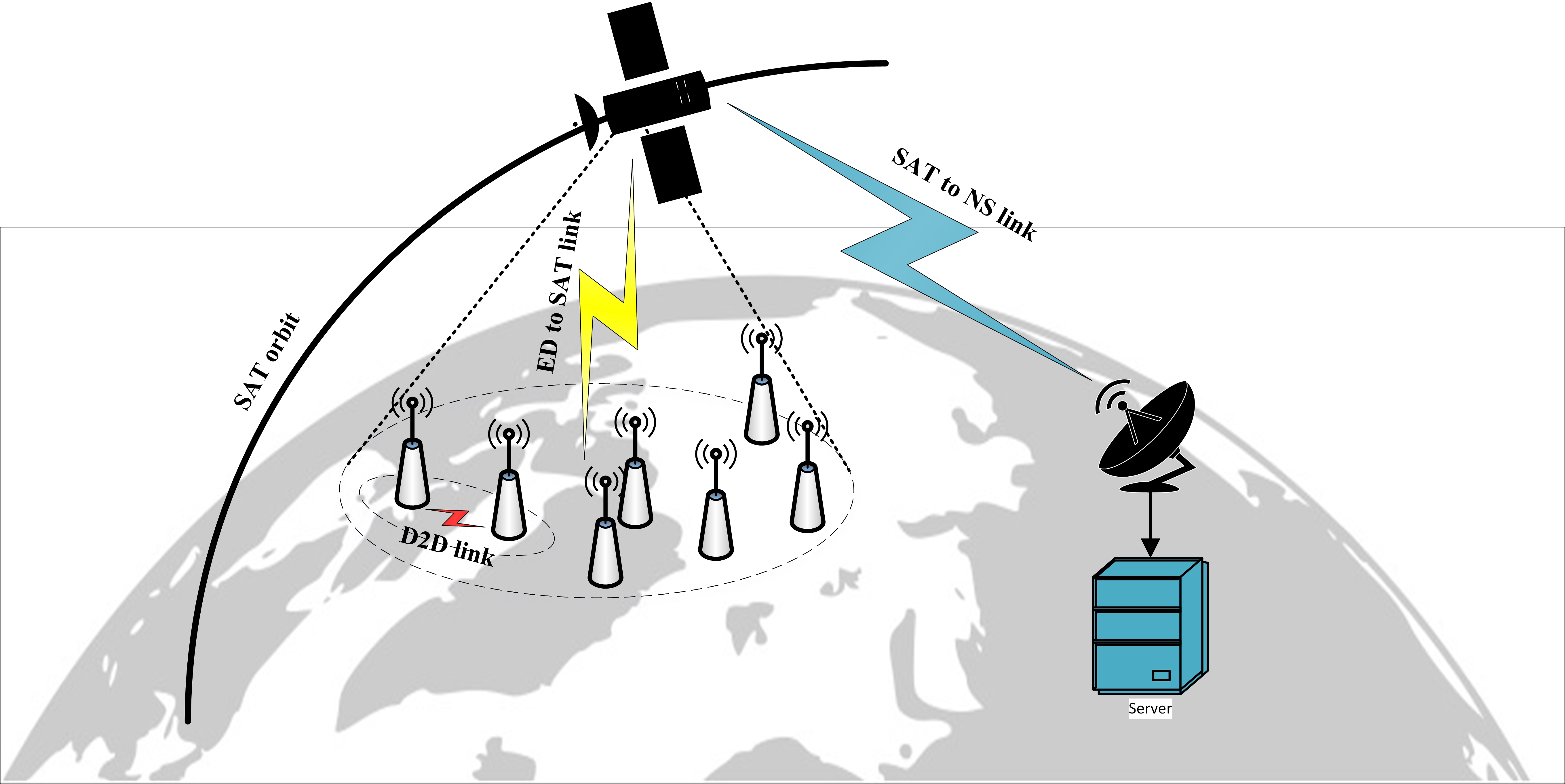}
  \caption{D2D-aided LR-FHSS network.}
  \label{LoRa_FHSS}
\end{figure}

To calculate the area covered by a LEO satellite with the speed of $\nu$, we present Fig. \ref{vert} to show a vertical view of satellite footprint movement on the surface of Earth in the duration of $T$. As is apparent, the covered area of $\mathcal{F}\in\mathbb{R}^2$ is equal to $|\mathcal{F}|=2\mathcal{R}_{\rm s}\nu T+\pi \mathcal{R}_{\rm s}^2$.

\begin{figure}[t!]
  \centering
  \includegraphics[scale=0.3]{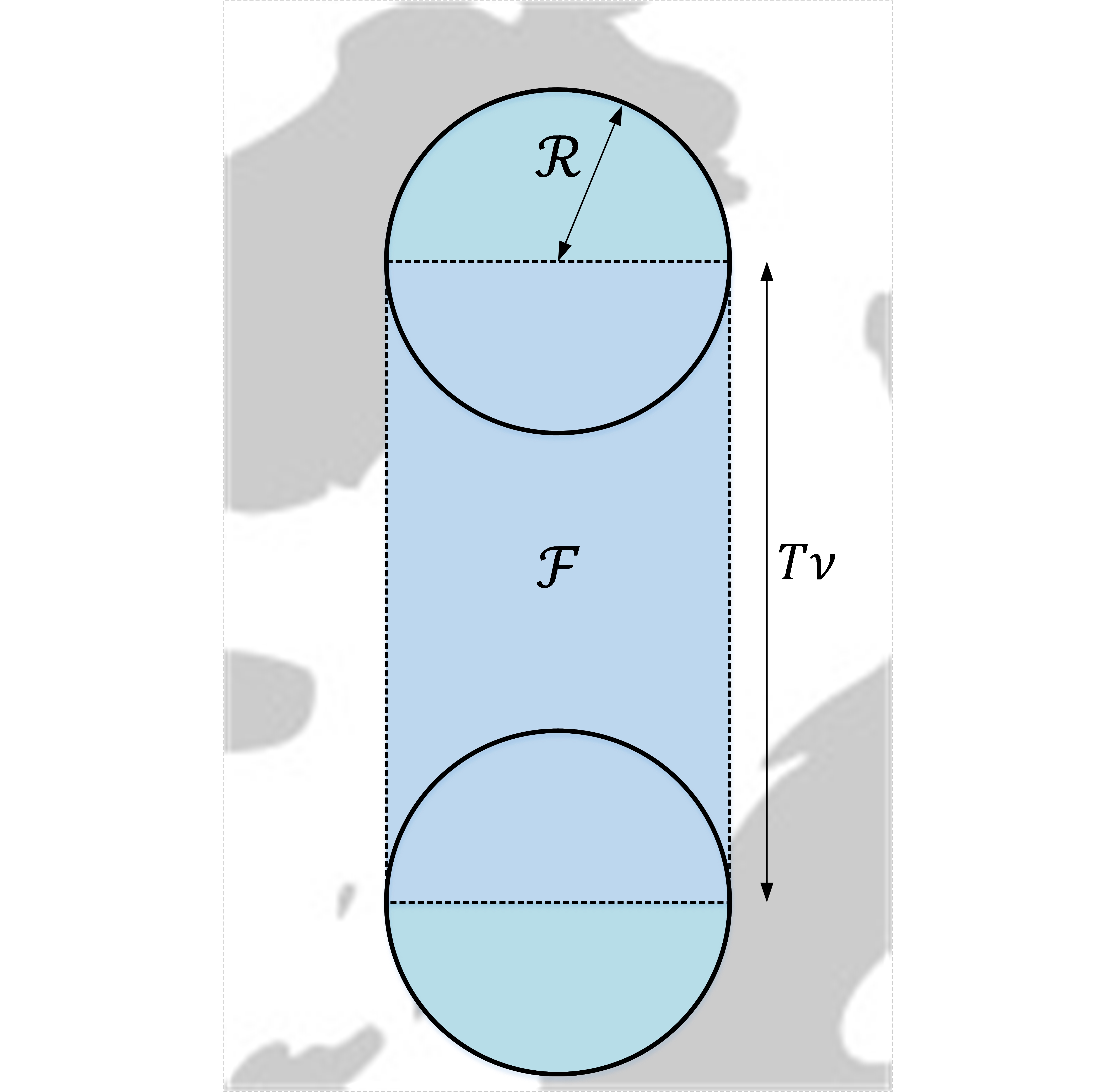}
  \caption{Vertical view of LEO satellite covered area in duration $T$.}
  \label{vert}
\end{figure}

Assume that we have $N_{\rm u}$ EDs uniformly distributed inside $\mathcal{F}$ according to a Poisson point process (PPP), and each ED generates a packet during the time-slot $T$ and transmits it at a random time instant. Therefore, the inter-arrival times for successive packets follow an exponential distribution with the mean of $t_{\rm{ave}}=T/(N_{\rm t} N_{\rm u})$. Considering the random nature of frequency hopping procedure, there might be co-channel interference between a desired ED, called ${\rm{ED}}_0$ hereinafter, and other devices transmitting inside the ${\rm{ToA_p}}$ of ${\rm{ED}}_0$. Consequently, the received signal at the IoT GW installed on the satellite in presence of $I_{\rm ci}$ co-channel interfering devices can be expressed as \cite{Santa_2020}:

\begin{equation}
\label{rec_sig}
r=\sqrt{P g(\alpha_0)} h_0 s_0+\sum_{i=1}^{I_{\rm ci}}{\sqrt{Pg(\alpha_i)} h_i s_i}+w,
\end{equation}
where $h_\ell$, $s_\ell$, and $g(\alpha_\ell)$, $\ell=0,1,\ldots, I_{\rm ci}$, represent the fading channel gain, modulated signal, and path loss at satellite elevation angle of $\alpha_\ell$ with respect to the $\ell$th ED, respectively. The term $w$ denotes AWGN with zero mean and variance $\sigma^2 = -174 + {\rm{NF}} + 10\log_{10}{B_{\rm{OBW}}}$ (dBm), where ${\rm NF}$ and $B_{\rm{OBW}}$ are the receiver's noise figure and occupied bandwidth \cite{17}. Moreover, the effective received power $P$ is given as $P=P_{\rm{t}} G_{\rm{t}} G_{\rm{r}}$, where $P_{\rm{t}}$, $G_{\rm{t}}$, and $G_{\rm{r}}$ represent the transmitted power, transmitter's antenna gain, and receiver's antenna gain, respectively. In the following, we present the exploited models for DtS channel fading and path loss.

\subsubsection{DtS channel fading}
The fading channel gain is characterized as a complex random variable $h_\ell=|h_\ell|\exp(j\angle{h_\ell})$. The elevation angle of the satellite and the Doppler shift are both taken into account by the various models designed for satellite communication channels \cite{13,14,15}. As stated in \cite{Saeed_2020}, the DtS channel models in the literature can be mainly categorized into static and dynamic channel models. In static channel models, the signal envelope distribution can be represented by a single constant-time distribution. On the opposite, the dynamic models are based on Markov chains, with several DtS channel states that each corresponds to various propagation settings. To be able to obtain a closed-form expression for the outage probability of the D2D-aided LR-FHSS system, we select the static channel model presented in \cite{13} which can be interpreted as a Rice fading channel with fluctuating (e.g., random) line-of-sight (LoS) component. Therefore, the shadowed-Rice probability density function (PDF) for the fading channel gain envelope $|h_\ell|$ can be expressed as \cite{13}:

\begin{IEEEeqnarray}{rCl}
\label{fading_model}
\mathfrak{f}_H(|h_\ell|) & =& \left(\frac{2{b_0}m}{2{b_0}m+\Omega}\right)^m \frac{|h_\ell|}{b_0} \exp\left(-\frac{|h_\ell|^2}{2b_0}\right) \nonumber \\
&&\times {{}_{1}F_{1}}\left[ m,1,\frac{\Omega |h_\ell|^2}{2{b_0}(2{b_0}m+\Omega)}\right],
\label{eq:dont_use_multline}
\end{IEEEeqnarray}
where $2b_0$ is the average power of the scattered component, $m$ is the Nakagami parameter, and $\Omega$ is the average power of the LoS component. Moreover, ${{}_{1}F_{1}}(.,.,.)$ is the confluent hypergeometric function \cite{Int_2014}. It is worth noting that unlike the Nakagami model, $m$ here changes over the range of $m\geq 0$. In fact, for $m=0$, (\ref{fading_model}) reduces to Rayleigh and for $m=\infty$, it becomes Rician. Based on the measurements and results provided in \cite{13}, we can model three different shadowing environments as presented in Table \ref{ShEnv}.  

\begin{table}[t!]
\caption{Shadowed-Rice parameters for different environments \cite{13}.}
\label{ShEnv}
\centering
\begin{tabularx}{0.49\textwidth}{Xccc}
\toprule[1.0pt]
Environmental conditions & $b_0$ & $m$ & $\Omega$ \\
\midrule[1.0pt]
Infrequent light shadowing & $0.158$ & $19.4$ & $1.29$ \\
Frequent heavy shadowing & $0.063$ & $0.739$ & $8.97\times 10^{-4}$ \\
Average shadowing & $0.126$ & $10.1$ & $0.835$ \\
\bottomrule[1.0pt]
\end{tabularx}
\end{table}

\subsubsection{Path-loss model}
We apply the model in \cite{16} to calculate the path loss, which concentrates on the propagation characteristics of satellite communication systems operating at $900-2100$ MHz in non-geostationary orbits. The selected path-loss model is specifically for rural shadowed locations with $99\%$ probability and is given as \cite{16}:

\begin{IEEEeqnarray}{rCl}
\label{path_loss}
g(\alpha)&=&32.44+20\log d(\alpha) +20\log f+L_{{\rm{air}}}(\alpha) \nonumber\\
&&+L_{{\rm{rain}}}(\alpha)+L_{{\rm{tree}}}(\alpha)+L_{{\rm{fog}}}+L_{{\rm{iono}}}+L_{{\rm{frad}}},
\label{eq:dont_use_multline}
\end{IEEEeqnarray}
where $d(\alpha)$, $f$, $L_{\rm{air}}$, $L_{\rm{rain}}$, $L_{\rm{tree}}$, $L_{\rm{fog}}$, $L_{\rm{iono}}$, and $L_{\rm{frad}}$ denote distance to the satellite (in km), operation frequency (in MHz), absorption of the atmosphere, attenuation caused by rain (absorption and scattering), tree, cloud or fog, ionosphere, and polarization effects, respectively. The distance $d(\alpha)$ is calculated as:

\begin{equation}
\label{Distance}
d(\alpha)=R_{\rm{e}}\left[\sqrt{\left(\frac{H_{\rm{s}}+R_{\rm{e}}}{R_{\rm{e}}}\right)^2-\cos^2(\alpha)}-\sin(\alpha)\right],
\end{equation}
where $R_{\rm{e}}=6378$ km is the Earth radius. The other path-loss parameters can be set as follows (all in dB) \cite{16}: $L_{\rm{air}}(\alpha)=0.1(1+\cos{(\alpha)})$, $L_{\rm{rain}}(\alpha)<0.1$, $L_{{\rm{fog}}}=0$, $L_{{\rm{iono}}}+L_{{\rm{frad}}}=3$, and

\begin{IEEEeqnarray}{rCl}
\label{tree}
L_{{\rm{tree}}}\left(\alpha\right)&=&\Bigg{[}25.8\exp{\left(-1.1\alpha\frac{1.57}{90}\right)}\nonumber\\
&&+1.5\cos{\left(\alpha\frac{3.937}{90}\right)}\Bigg{]}\sqrt{\frac{f}{900}}.
\label{eq:dont_use_multline}
\end{IEEEeqnarray}

\subsubsection{Average number of interfering devices}
First, assume that ${\rm{ED}}_0$ starts its packet transmission at time instant $t_0$ and selects the group $G_0$ out of $G$ available frequency groups to generates its frequency hopping sequence with the length of $N_{\rm{HDR}}+N_{\rm{PL}}$. The average number of EDs that start their packet transmission in the range of $[t_0-{\rm{ToA_p}},t_0+{\rm{ToA_p}}]$ can be formulated as: 

\begin{equation}
\label{Inter}
I=\left\lceil \frac{2{\rm ToA}_{\rm P}}{t_{\rm ave}} \right\rceil-1.
\end{equation}

Note that these are the EDs that can potentially cause co-channel interference to ${\rm{ED}}_0$ if they select the same frequency group as ${\rm{ED}}_0$ for hopping sequence generation (see Fig. \ref{Int}, (a)). Because the ToAs for header and payload fragments differ, we must define the average numbers of interfering devices for header and payload transmissions separately. To this end, consider that $I'$ ($I'=0,1,\dots,I$) out of these $I$ time-domain interfering devices choose $G_0$. By defining $I'=0$ as ``no interference" case and excluding it from our assumption, for now, we define the average inter-arrival times of header and payload fragments in the duration of $2{\rm{ToA_p}}$ as follows:

\begin{equation}
\label{t_HDR}
t_{\rm HDR}(I')=\frac{2{\rm ToA}_{\rm P}}{I'\times N_{\rm HDR}},
\end{equation}

\begin{equation}
\label{t_PL}
t_{\rm PL}(I')=\frac{2{\rm ToA}_{\rm P}}{I'\times N_{\rm PL}}.
\end{equation}

During a header transmission, the ${\rm{ED}}_0$ can experience interference from other EDs' header (left of Fig. \ref{Int}, (b)) or payload parts (right of Fig. \ref{Int}, (b)). Hence, the average number of header interference can be formulated as:

\begin{equation}
\label{I_HDR}
I_{\rm HDR}(I')=\frac{2T_{\rm HDR}}{t_{\rm HDR}(I')}+\frac{T_{\rm HDR}+T_{\rm PL}}{t_{\rm PL}(I')}.
\end{equation}
Similarly, during a payload transmission, interference can occur due to other EDs' payload (left of Fig. \ref{Int}, (c)) or header parts (right of Fig. \ref{Int}, (c)). Therefore, the average number of interfering devices from payload parts is
\begin{equation}
\label{I_PL}
I_{\rm PL}(I')=\frac{2T_{\rm PL}}{t_{\rm PL}(I')}+\frac{T_{\rm PL}+T_{\rm HDR}}{t_{\rm HDR}(I')}.
\end{equation}

\begin{figure}[t!]
  \centering
  \includegraphics[width=\linewidth]{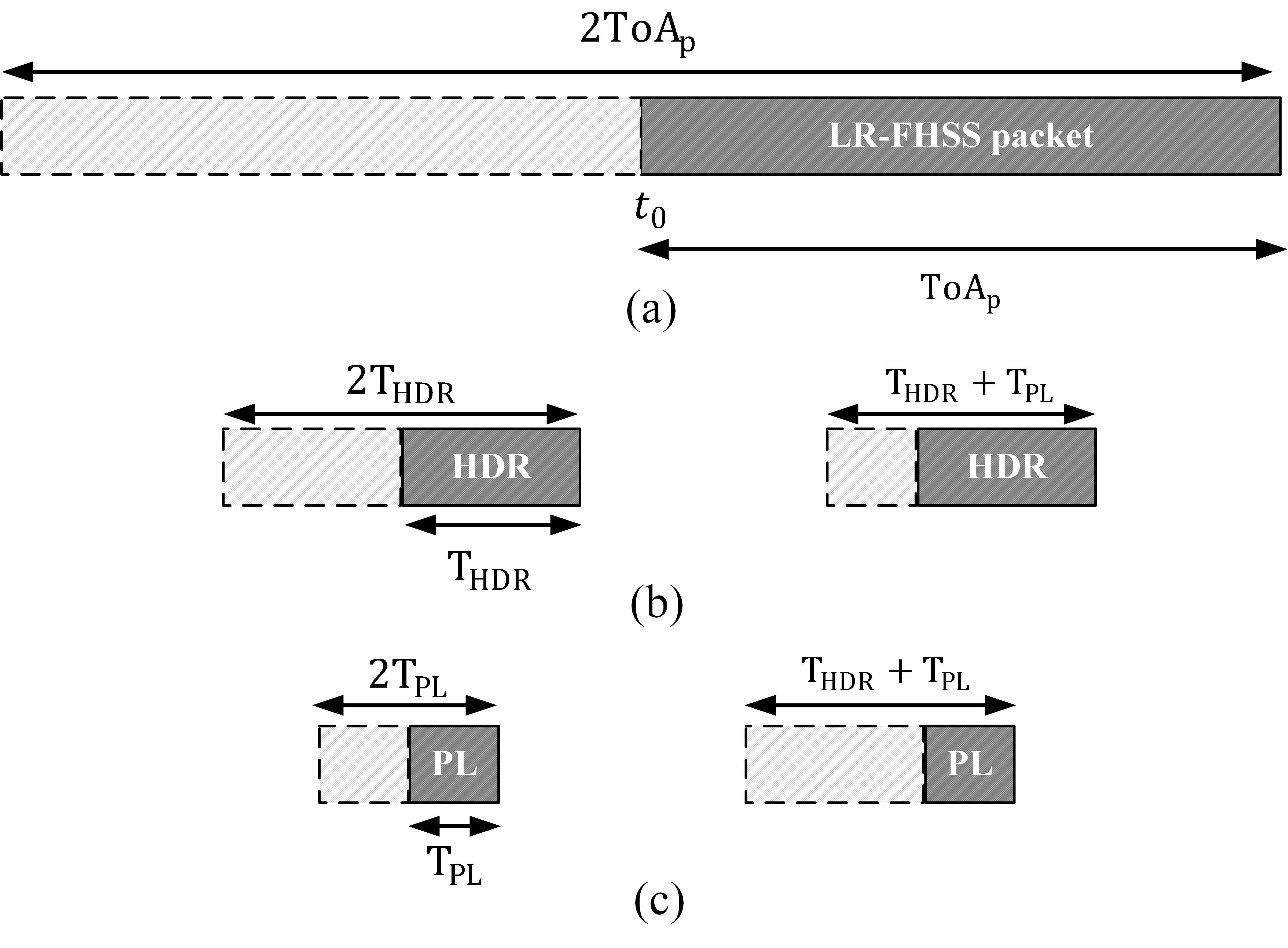}
  \caption{Time range of interference occurrence for LR-FHSS: (a) packet transmission, (b) header transmission, and (c) payload transmission.}
  \label{Int}
\end{figure}

\section{LR-FHSS Outage Probability Analysis}
In this section, a closed-form expression for the outage probability of the LR-FHSS scheme in the shadowed-Rice fading environment is derived. 

Considering the LR-FHSS transmission framework, first we define three events, i.e., outage, disconnection, and collision as follows:

\begin{definition}
\label{def1}
\textbf{Outage.} The outage event in LR-FHSS scheme is defined as the case in which at least one of these two conditions occur \cite{10}: (a) all $N_{\rm{HDR}}$ header replicas are lost or (b) a minimum of one header is decoded successfully, but more than $(1-\kappa)N_{\rm{PL}}$ payload fragments are lost, where $\kappa$ is the forward error correction (FEC) code rate.
\end{definition} 

\begin{definition}
\textbf{Disconnection.} The disconnection event is defined as the case in which the packet's signal-to-noise ratio (SNR) is lower than a certain threshold $\psi$.
\end{definition}

\begin{definition}
\label{def2}
\textbf{Collision.} The collision event is defined as the case in which (a) a packet fragment of the desired device collides with one or more other devices' packet fragments in time and frequency and (b) the power differential between the desired packet element and the sum of interfering fragments is smaller than a set threshold $\delta$.
\end{definition}

Mathematically, the probability of disconnection for one packet element, regardless of belonging to header or payload, can be expressed as:

\begin{equation}
\label{P_disc}
P_{\rm disc}=\Pr{\left\{\frac{P\left|h_0\right|^2g\left(\alpha_0\right)}{\sigma^2}\le\psi \Bigg{|}g(\alpha_0)\right\}}.
\end{equation}

\begin{lemma}
In the shadowed-Rice fading environment, the probability of disconnection can be written as:

\begin{equation}
\label{P_disc_close}
P_{\rm{disc}}=A \sum\limits_{n=0}^{\infty} \frac{(m)_n}{n! n!} C(n) (\frac{1}{B})^{n+1} \gamma\left(n+1,{B\frac{\psi\sigma^2}{Pg_0}}\right).
\end{equation}
where $(m)_n$ is the Pochhammer symbol defined as $(m)_n=m(m+1)\dots (m+n-1)$ and $\gamma(a,x)$ denotes the unnormalized incomplete Gamma function as:

\begin{equation}
\label{gamma}
\gamma(a,x)=\int_{0}^{x} e^{-t} t^{a-1} \mathrm{d}t.
\end{equation}
Moreover, we have:

\begin{equation}
\label{5}
A=\left(\frac{2{b_0}m}{2{b_0}m+\Omega}\right)^m \frac{1}{2b_0},
\end{equation}

\begin{equation}
\label{55}
B=\frac{1}{2b_0}, 
\end{equation}

\begin{equation}
\label{6}
C(n)=\left[\frac{\Omega}{2{b_0}(2{b_0}m+\Omega)}\right]^n.
\end{equation}

\end{lemma}

\begin{proof}
Please refer to Appendix A.
\end{proof}

Note that a packet fragment can be lost either due to disconnection or collision or both. On the other hand, we have the average number of $\lceil I_{\rm HDR}(I')\rceil$ time-domain interfering deices in the duration of one header fragment. By using a binomial distribution, we sum over all the cases that $k$ of this $\lceil I_{\rm HDR}(I')\rceil$ EDs select the same frequency as the desired header fragment out of $S$ available frequencies with a probability of $1/S$. Therefore, the probability of condition (a) in Definition \ref{def1} due to $I'$ EDs selecting same group as ${\rm{ED}}_0$ can be computed as:

\begin{IEEEeqnarray}{rCl}
\label{P_HDR}
P_{\rm HDR}(I')&=&\Bigg{[}P_{\rm disc}+\left(1-P_{\rm disc}\right)\sum_{k=1}^{\left\lceil I_{\rm HDR}(I')\right\rceil}\left(\begin{matrix}\left\lceil I_{\rm HDR}(I')\right\rceil\\k\\\end{matrix}\right)\nonumber \\
&&\times \left(\frac{1}{S}\right)^k\left(\frac{S-1}{S}\right)^{\left\lceil I_{\rm HDR}(I')\right\rceil-k} P_{\rm cap}(k)\Bigg{]}^{N_{\rm HDR}}.\nonumber \\
\label{eq:dont_use_multline}
\end{IEEEeqnarray}
In the above expression, $P_{\rm cap}(k)$, referred to as ``capture failure probability", represents the probability of condition (b) in Definition \ref{def2}. By incorporating the capture effect, i.e., receiving the strongest packet fragment successfully and considering other packet fragments as interference, the capture failure probability probability is expressed as:

\begin{equation}
\label{P_cap}
P_{\rm cap}(k)=\Pr{\left\{\frac{\left|h_0\right|^2g_0}
{\sum_{m=1}^{k}\left|h_m\right|^2g_m}\le\delta\Bigg{|}g_0,g_1,\dots,g_k\right\}}.
\end{equation}
Note that hereinafter, we use $g_m$ instead of $g(\alpha_m)$ for the sake of simplifying the notation.

\begin{lemma}
The capture failure probability can be expressed as follows:

\begin{IEEEeqnarray}{rCl}
\label{P_cap_close}
P_{\rm cap}(k)&=& 1 - AD\sum\limits_{i=0}^{\infty}c_i  \sum\limits_{n=0}^{\infty}\frac{(m)_n}{n!n!}C(n) \times (\frac{1}{B})^{n+1} \Gamma (n+1)\nonumber \\
&& + AD\sum\limits_{i=0}^{\infty}c_i \sum\limits_{u=0}^{k+i-1}\frac{1}{u!} (\frac{g_0}{\alpha\delta})^u\times \Gamma(u+1)\times {B'} ^{-(u+1)}\nonumber \\
&&\times {{}_{2}F_{1}} (m,u+1,1,\frac{C(1)}{B'}),
\label{eq:dont_use_multline}
\end{IEEEeqnarray}
where $\alpha$, $\Gamma(.)$, and ${{}_{2}F_{1}}(.,.,.,.)$ represent an arbitrary positive number in the range of $[0,4\min\limits_{m}\{b_0 g_m\}]$, the Gamma function \cite{Int_2014}, and the Gauss hypergeometric function \cite{Int_2014}, respectively. Also, we have:

\begin{equation}
\label{555}
B'=B+\frac{g_0}{\alpha \delta}, 
\end{equation}
and
\begin{equation}
\label{66}
D=\alpha^k\prod\limits_{i=1}^{k}\frac{(2b_0 g_i)^{m-1}}{\left(2b_0 g_i+\frac{\Omega g_i}{m}\right)^m}. 
\end{equation}

\end{lemma}

\begin{proof}
Please refer to Appendix B.
\end{proof}

The next step toward the calculation of LR-FHSS outage probability is to obtain the probability of condition (b) in Definition \ref{def1} using a similar approach to $P_{\rm HDR}(I')$. To this end, first we focus on one payload fragment. The probability of a single payload fragment (SPF) being lost due to due to $i_0$ EDs selecting same group as ${\rm{ED}}_0$ can be written as

\begin{IEEEeqnarray}{rCl}
\label{P_SFP}
P_{\rm SPF}(I')&=&P_{\rm disc}+\left(1-P_{\rm disc}\right)\sum_{k=1}^{\left\lceil I_{\rm PL}(I')\right\rceil}{\left(\begin{matrix}\left\lceil I_{\rm PL}(I')\right\rceil\\k\\\end{matrix}\right)\left(\frac{1}{S}\right)^k}\nonumber \\
&&\times \left(\frac{S-1}{S}\right)^{\left\lceil I_{\rm PL}(I')\right\rceil-k}P_{\rm cap}(k).
\label{eq:dont_use_multline}
\end{IEEEeqnarray}
Consequently, the probability of condition (b) in Definition \ref{def1} can be formulated as:

\begin{IEEEeqnarray}{rCl}
\label{P_PL}
P_{\rm PL}(I')&=&\sum_{m=\omega}^{N_{\rm PL}}\left(\begin{matrix}N_{\rm PL}\\m\\\end{matrix}\right)P_{\rm SPF}(I')^m \times \left[1-P_{\rm SPF}(I')\right]^{N_{\rm PL}-m},\nonumber\\
\label{eq:dont_use_multline}
\end{IEEEeqnarray}
where $\omega = \left\lceil\left(1-\kappa\right)N_{\rm PL}\right\rceil$. 

Finally, by combining (\ref{P_HDR}) and (\ref{P_PL}), the outage probability of the LR-FHSS system is obtained as:

\begin{IEEEeqnarray}{rCl}
\label{P_OUT}
O_{\rm L}&=&\sum_{I'=1}^{I}\left(\begin{matrix}I\\I'\\\end{matrix}\right)\left(\frac{1}{G}\right)^{I'}\left(\frac{G-1}{G}\right)^{I-I'}\nonumber \\
&&\times \left[P_{\rm H}(I')+\left(1-P_{\rm H}(I')\right)P_{\rm PL}(I')\right]+P_{\rm{NI}},
\label{eq:dont_use_multline}
\end{IEEEeqnarray}
where $P_{\rm{NI}}$ denotes the probability of packet loss in ``no interference" case. To calculate this term representing $I'=0$ situation, the only reason that might result in a packet fragment loss is noise. Therefore, we have:

\begin{IEEEeqnarray}{rCl}
\label{P_NI}
P_{\rm NI}&=& \left(\frac{G-1}{G}\right)^I \Bigg{[}P_{\rm disc}^{N_{\rm HDR}} + (1-P_{\rm disc}^{N_{\rm HDR}}) \nonumber \\
&&\times \sum_{m=\omega}^{N_{\rm PL}}\left(\begin{matrix}N_{\rm PL}\\m\\\end{matrix}\right)P_{\rm disc}^m \times \left[1-P_{\rm disc}\right]^{N_{\rm PL}-m}\Bigg{]}.
\label{eq:dont_use_multline}
\end{IEEEeqnarray}

\section{D2D-aided LR-FHSS Scheme}
In this section, we present the D2D-aided LR-FHSS scheme along with its network management aspects and outage probability derivation. 

Our proposed scheme contains three main sessions, i.e. D2D, original data transmission, and parity data transmission \cite{Oliveira_2022}:

\begin{enumerate}
\item D2D session: In this session, two cooperating EDs exchange their data using LoRa modulation with a limited coverage.
\item Original data transmission: After successfully completing the data transfer using D2D link, in this session, both EDs transmit their packets using LR-FHSS scheme to the IoT GW on the LEO satellite.
\item Parity data transmission: Finally, both EDs transmit their parity signals obtained as linear combinations of original packets via LR-FHSS. 
\end{enumerate}

Before providing the analysis for the outage performance, we present the network management scheme of the D2D-aided LR-FHSS framework based on the practical features and specifications available in LoRaWAN. 

\subsection{Network Management of D2D-aided LR-FHSS}
The EDs are assumed to be Semtech SX1261/2 \cite{SX1261} and SX1268 \cite{SX1268} modems (all class A devices) which are fully LR-FHSS compatible and allowed to transmit data at any time followed by two receiving windows (RW) for downlink communications including NS commands. Furthermore, the D2D session should be managed in a network-assisted manner as shown in \cite{Mikhaylov_2017} meaning that the D2D link establishment should be performed under the supervision of NS. This can lead to low overheads and lower possibility of negative effects on the network operation. Therefore, for establishing the D2D link, the NS can initiate the D2D session by a downlink command between two EDs at time $t_{\rm init}$ based on the satellite trajectory.

On the other hand, the NS has means to enable localization of each ED \cite{Mikhaylov_2017}. So, it is fair to assume that the NS has knowledge of each ED's location. We consider that the data exchange between cooperating EDs is established using conventional LoRa modulation. Knowing that SF12 provide the longest range in LoRa-based communications, we define a distance threshold, $d_{\rm max}$, within which two EDs can communicate with each other via LoRa SF12. In fact, $d_{\rm max}$ can be translated into the maximum of the required distance between an ED and other EDs for being able to establish the D2D link. Note that the distance between two cooperating EDs can be much lower than $d_{\rm max}$ and NS will assign lower SFs for their D2D session using the adaptive data rate (ADR) feature of LoRaWAN. Based on these features, NS can perform clustering on $N_{\rm u}$ EDs and assign a cluster to each pair of two EDs with a distance lower than $d_{\rm max}$. Assume that we obtain $N_{\rm c}$ clusters using this approach. Therefore, the remaining $N_{\rm s}=N_{\rm u}-2N_{\rm c}$ EDs are considered as single EDs and does not receive a D2D establishment command. Similar to the approach in \cite{Hoeller_2018,Santa_2020}, these $N_{\rm s}$ EDs use retransmission (RT) scheme with two transmission per timeslot for exploiting diversity.

For our proposed network, we present a D2D-aided LR-FHSS protocol similar to the LoRaWAN-D2D protocol presented in \cite{Mikhaylov_2017}. Consider two EDs, ${\rm{ED}}_{{\rm c}_{i,1}}$ and ${\rm{ED}}_{{\rm c}_{i,2}}$, that are assigned to the cluster ${\rm c}_i$. The D2D session is initiated at $t_{{\rm c}_i}$ via an NS downlink command as indicated in Fig. \ref{D2D}. Afterward, one ED starts its packet transmission using LoRa modulation (initiator mode) while the other ED enables its receiving window (RW) (scanner modes). After completion of the first data exchange, EDs switch modes and the data transfer continues. At the end of the D2D session, each ED waits a random time and transmits its original data to the GW. Then, they wait for another random period and transmit their parity signal obtained as linear combination of original packets to the GW. It should be noted that the detailed time scheduling of the presented protocol is out of the scope of this work. However, these timing requirements can be satisfied using an effective time scheduling algorithm taking into account the satellite trajectory such as the one suggested in \cite{Afhamisis_2022}.

After receiving all of a cluster's packets, the GW first tries to decode each packet separately using the standard LR-FHSS method. The NS is then sent the decoded packets exploiting methods like Gaussian elimination to retrieve both original frames from any two of the four received ones in a subsequent stage. Therefore, aside from the extra network coding/decoding modules needed for the D2D-aided LR-FHSS scheme, the rest of the transmission process follows the standard LR-FHSS principles.  

\begin{figure}[t!]
  \centering
  \includegraphics[width=\linewidth]{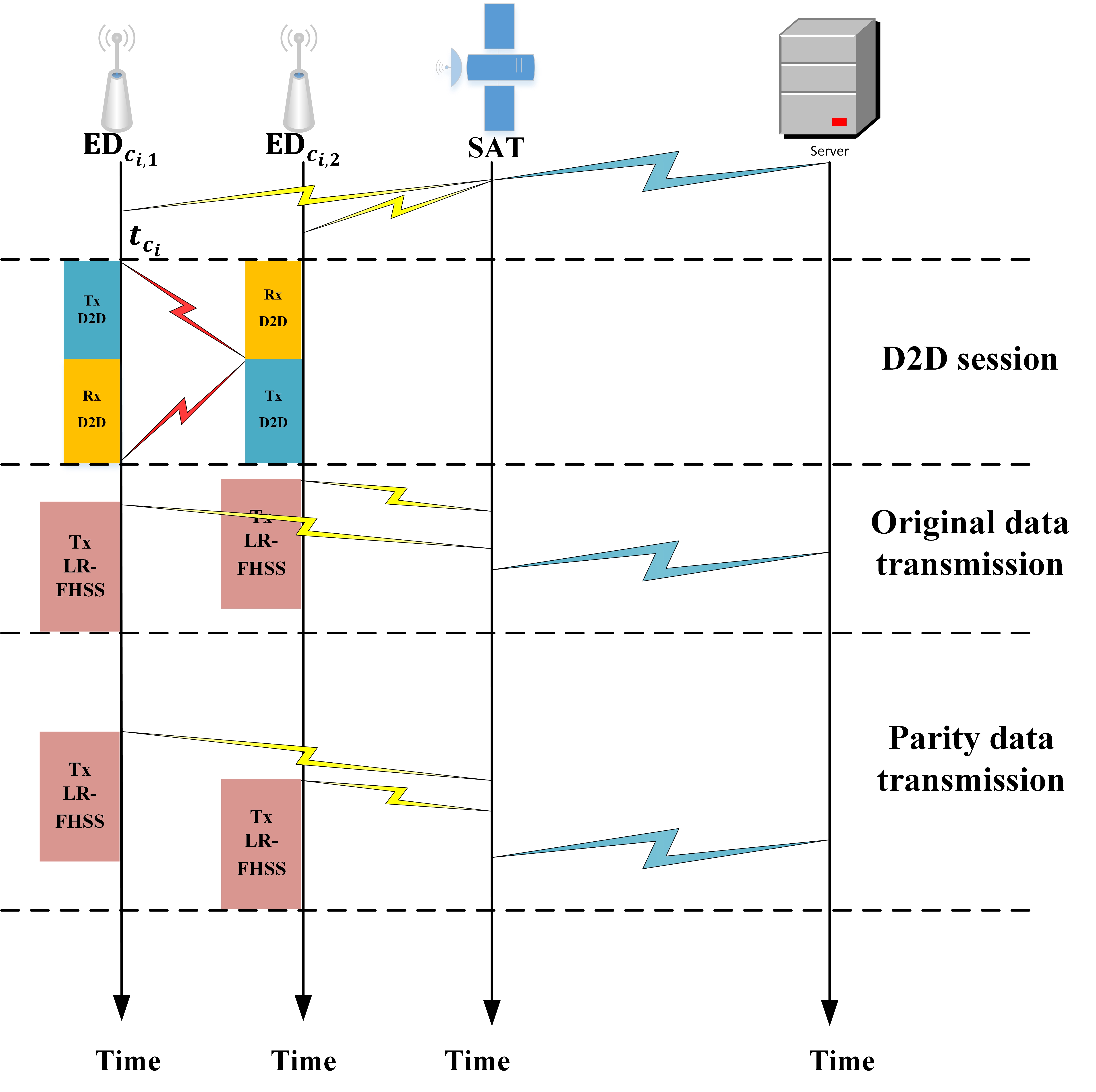}
  \caption{D2D-aided LR-FHSS time scheduling.}
  \label{D2D}
\end{figure}

\subsection{D2D-aided LR-FHSS Outage Probability Analysis}
First, we need to calculate the probability of successful D2D communication. Considering the random distribution of EDs, the probability of having an ED ($\rm ED_{ne}$) with a distance lower than $d_{\rm max}$ to ${\rm ED}_0$ can be expressed as (neighboring probability) \cite{Oliveira_2022}:

\begin{equation}
\label{P_neigh}
P_{\rm ne}=1-\exp (-\rho A_{\rm D2D}),
\end{equation}
where $\rho$ is the intensity of the PPP and $A_{\rm D2D}=\pi d_{\rm max}^2$. Therefore, the probability of successful D2D communication is written as:

\begin{equation}
\label{P_D2D}
P_{\rm D2D}=(1-P_{\rm LoRa})P_{\rm neigh},
\end{equation}
where $P_{\rm LoRa}$ is the outage probability of packet exchanging procedure between ${\rm ED}_0$ and $\rm ED_{ne}$ using LoRa communication. It is worth mentioning that if the D2D communication fails, both EDs will switch to RT scheme by setting the linear coefficients relating to the cooperating partner to zero when generating the parity packet. This makes the transition from network coding to RT scheme very simple to perform in practice.

Upon successful reception of the original packets $o_{\rm ne}$ and $o_0$ by ${\rm ED}_0$ and $\rm ED_{ne}$, respectively, they will generate parity packets using the network coding approach presented in \cite{Rebelatto_2012,Xiao_2010}. In an uplink network coding strategy, nodes take advantage of the wireless channel's broadcast characteristics to communicate with and support one another. A given node's relayed data is a linear combination of its information and the information from other cooperating partners conducted over a non-binary finite field ${\rm GF}(4)$. The coefficients for the linear combinations are selected to result in a maximum-distance separable (MDS) code for which maximum achievable diversity is guaranteed \cite{Rebelatto_2012}. Therefore, parity packets of $p_0 = o_0 \boxplus o_{\rm ne} $ and $p_{\rm ne} = o_0\boxplus 2 \boxtimes o_{\rm ne}$ are generated by ${\rm ED}_0$ and $\rm ED_{ne}$, respectively, in which $\boxplus$ and $\boxtimes$ are summation and multiplication operations, respectively, over ${\rm{GF}}(4)$. 

By defining the received vector corresponding to the cluster of ${\rm ED}_0$ and $\rm ED_{neigh}$ as $\mathbf{r}=[o_0,o_{\rm ne},p_0,p_{\rm ne}]$, we present the definition of D2D-aided LR-FHSS outage as follows:

\begin{definition}
\label{def3}
The packet $o_0$ of ${\rm{ED}}_0$ will be in outage, if the D2D is established, and $o_0$ and at least two of the $o_{\rm ne}$, $p_0$ and $p_{\rm ne}$ are lost, or D2D cannot be established, and both transmission and retransmission of $o_0$ are lost.
\end{definition}

Finally, based on Definition \ref{def3}, the outage probability of the D2D-aided LR-FHSS can be given as:

\begin{IEEEeqnarray}{rCl}
\label{P_out_D_Long}
O_{\rm D}&=&P_{\rm D2D}O_{{\rm L},o_0}\Big{[}O_{\rm L,o_{ne}}\left({1-O}_{{\rm L},p_0}\right)\left(1-O_{{\rm L},p_{{\rm ne}}}\right)\nonumber \\
&&+O_{{\rm L},o_{\rm ne}}O_{{\rm L},p_0}\left(1-O_{{\rm L},p_{\rm ne}}\right)+O_{{\rm L},o_{\rm ne}}O_{{\rm L},p_0}O_{{\rm L},p_{\rm ne}}\Big{]}\nonumber\\
&&+\left(1-P_{\rm D2D}\right)O_{{\rm L},o_0}^2.
\label{eq:dont_use_multline}
\end{IEEEeqnarray}

Considering that the satellite orbital height is much larger than the distance between two cooperating EDs, i.e., $H_{\rm s}\gg d_{\rm max}$, we can assume the distances from node to satellite for both EDs are approximately equal and all four packets in a cluster experience similar outage performance. Therefore, we can approximate $O_{\rm L}\approx O_{{\rm L},o_0}\approx O_{{\rm L},o_{\rm ne}}\approx O_{{\rm L},p_0}\approx O_{{\rm L},p_{\rm ne}}$. Hence, we obtain: 

\begin{IEEEeqnarray}{rCl}
\label{P_out_D}
O_{\rm D} \approx P_{\rm{D2D}}\left[O_{\rm L}(-2O_{\rm L}^{3}+3O_{\rm L}^{2})\right] + (1-P_{\rm{D2D}})O_{\rm L}^2.
\label{eq:dont_use_multline}
\end{IEEEeqnarray}

\section{Numerical Results}
Simulation results is presented in this section to evaluate the performance of the LR-FHSS and D2D-aided LR-FHSS schemes and verify the mathematical analysis presented throughout the paper. The parameters used in the simulation process are summarized in Table \ref{simpar}.

\begin{table*}[t!]
\caption{Simulation parameters.}
\label{simpar}
\centering
\begin{tabularx}{\textwidth}{Xcc}
\toprule[1.0pt]
LEO satellite coverage radius ($\mathcal{R}$) \cite{12} & $2,209$ km \\
Satellite orbital height ($H_{\rm s}$) \cite{12} & $780$ km \\
Satellite speed ($\nu$) \cite{12} & $7.4$ km/s \\
LR-FHSS transmitted power ($P_{\rm t}$) \cite{LR111} & $30$ dBm \\
LR-FHSS signal bandwidth ($B_{\rm OBW}$) \cite{9} & $488$ Hz\\
Transmit antenna gain ($G_{\rm t}$) \cite{12} & $2.5$ dBi \\
Gateway antenna gain ($G_{\rm r}$) \cite{12} & $22.6$ dBi \\
Carrier frequency & $905.4385$ MHz \\
Header duration ($T_{\rm HDR}$) \cite{9} & $233$ ms \\
Header replicas ($N_{\rm HDR}$) \cite{9} & $3$ (DR5) and $2$ (DR6)\\
Payload fragment duration ($T_{\rm PL}$) \cite{9} & $102$ ms \\
Payload & $30$ bytes \\
Payload fragments \cite{9} ($N_{\rm PL}$) & $5$ \\
LR-FHSS time on air & $1.209$ s (DR5) and $0.976$ s (DR6)\\
LR-FHSS receiver's SNR threshold ($\psi$) \cite{LR111} & $3.96$ dB \\
Sum-interference cancellation threshold ($\delta$) \cite{12} & $6$ dB \\ 
Noise figure (NF) & $6$ dB \\
Simulation time ($T$) & 291.1 s\\
\bottomrule[1.0pt]
\end{tabularx}
\end{table*}

Please not that the simulation time $T$, mentioned as the time slot in the equation (\ref{duty_cycle}) is set to be $291.1$ seconds in our simulations. This selection is based on coping with $1\%$ duty cycle constraint of LoRaWAN protocol for the worst-case scenario in the D2D-aided LR-FHSS scheme in which an ED can complete all three main sessions as discussed in Section~IV. In such a case, the duty cycle should include one LoRa D2D transmission (using SF12 since it has the largest ToA) and two LR-FHSS transmissions (using DR5). Therefore, we have a total time of $2{\rm ToA_{p (DR5)}}+{\rm ToA_{p (SF12)}}$ for an ED transmission. Using the $30$ bytes of payload, a bandwidth of $500$ kHz for LoRa communication, and a maximum duty cycle of $1\%$ in the United States region, we obtain $T=291.1$ s. It is also worth noting that for fairness in comparison between LR-FHSS and D2D-aided LR-FHSS, we assume the same time slot for LR-FHSS as for D2D-aided LR-FHSS.

It is also worth noting that for the sake of simulations, the location of each ED is generated uniformly according to a PPP over the satellite’s coverage area on Earth. This results in a random nature for the path-loss term $g_m$ as indicated by conditional probabilities of (\ref{P_disc}) and (\ref{P_cap}). Although in practice, we assume that the NS may have knowledge about the location of each ED, to capture the effect of this randomness in the simulation, the probabilities of $P_{\rm{disc}}$ and $P_{\rm cap}$ are obtained by averaging over $10000$ realizations of device locations.

According to (\ref{P_disc}), $P_{\rm disc}$ captures the effect of channel noise, channel fading, and path loss on the outage performance of the proposed scheme. In Fig. \ref{disc}, the probability of disconnection is illustrated for different shadowing conditions mentioned in Table \ref{ShEnv} to show the effect of fading environments. The series in (\ref{P_disc_close}) is calculated using its first $10$ terms for the analytical curve, and the numerical one is obtained by performing numerical integration on (\ref{ap1}). As can be seen, for a certain value of transmitted power, $P_{\rm disc}$ is heavily impacted by the shadowed-Rice fading parameters. The infrequent shadowing environment measurements in which the DtS communication exploits a relatively strong LoS link \cite{13} results in the best performance in terms of disconnection probability. However, in practice, this is not always the case. Therefore, for the rest of the simulation parts, we assume that the DtS link behaves according to the average shadowing environment in terms of channel fading.

\begin{figure}[t!]
  \centering
  \includegraphics[width=\linewidth]{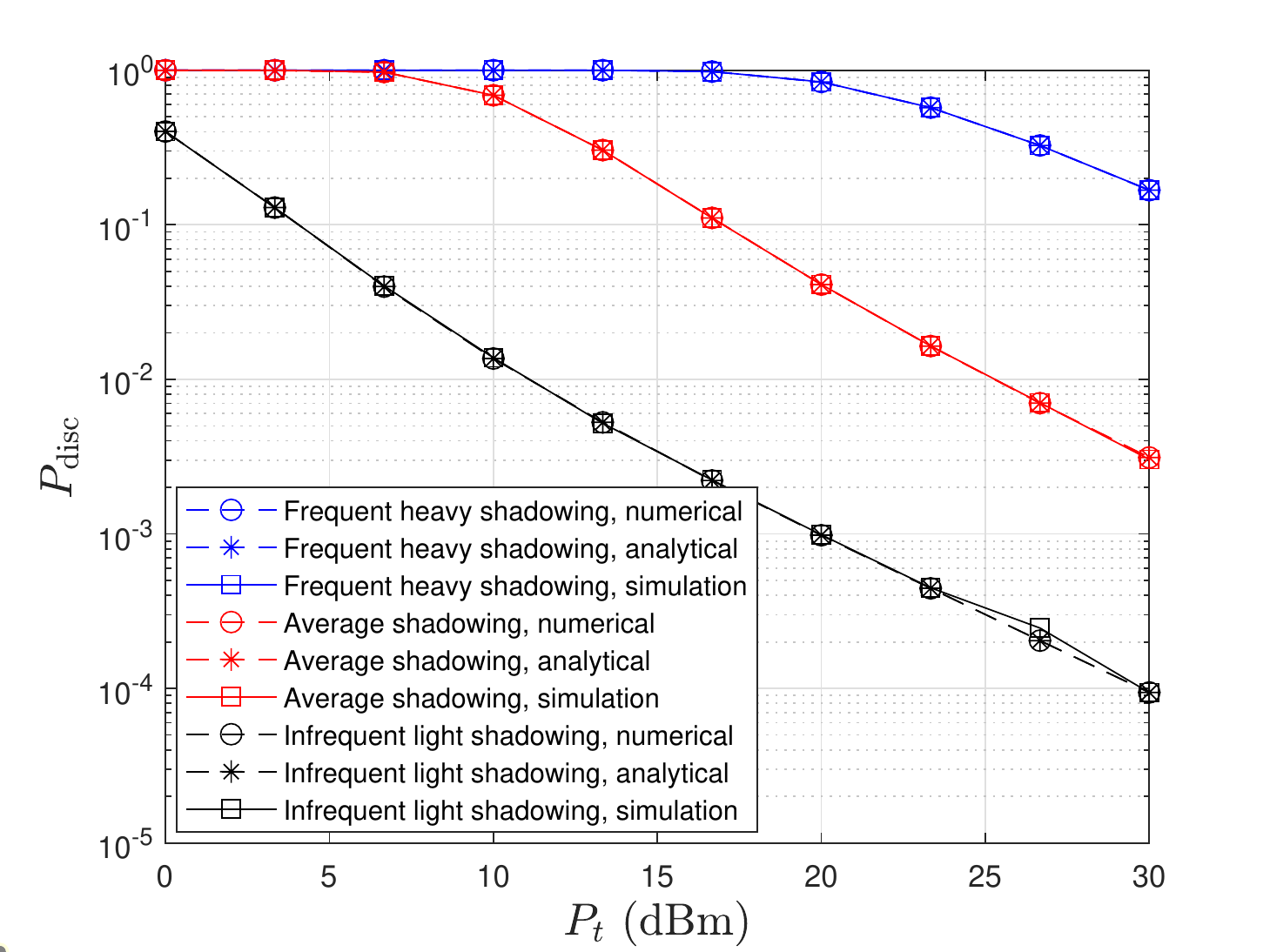}
  \caption{Probability of disconnection for three different shadowing environments as function of LR-FHSS transmission power.}
  \label{disc}
\end{figure}

Although the channel fading, noise, and path loss are encompassed by the disconnection probability, capture failure probability merely depicts the effect of interference as indicated in Fig. \ref{cap}. As can be seen, in the case of having more than $8$ interfering devices, the packet fragment will be dropped with a very high probability. However, it should be noted that the value of $6$ dB for $\delta$ is selected as in \cite{12}, since technical information about the LR-FHSS sum-interference cancellation threshold is not available to the author's best knowledge. Moreover, for plotting the analytical $P_{\rm cap}$, (\ref{P_cap_close}) is calculated using the first $10$ terms of the series $\sum_{n=0}^{\infty}(.)$. It is noteworthy that the parameter $\alpha$ controls the convergence of the series $\sum_{i=0}^{\infty}(.)$, and hence, the number of series terms needed to obtain the desired accuracy. We set $\alpha=3.9999 \times \min\limits_{m}\{b_0 g_m\}$ using trial and error resulting in a relatively fast convergence with the accuracy of $5\%$ compared to the results obtained by simulation.

\begin{figure}[t!]
  \centering
  \includegraphics[width=\linewidth]{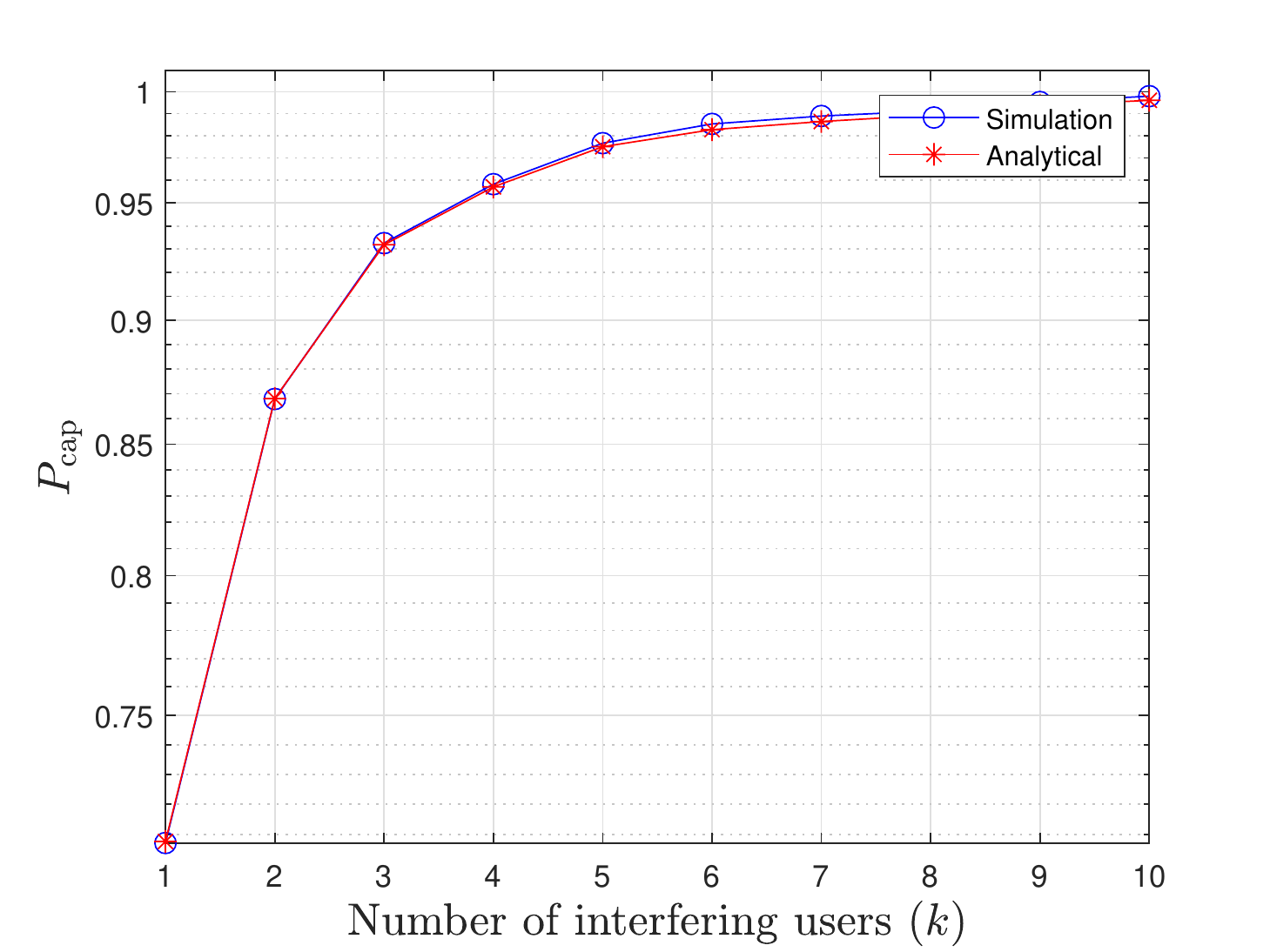}
  \caption{Probability of receiver’s failure in detecting the desired signal in the presence of interfering devices.}
  \label{cap}
\end{figure}

Considering the fact that the time scheduling scheme of LR-FHSS is based on LoRaWAN unslotted ALOHA, the first condition that should be satisfied for an ED to be considered as a potential co-channel interference is to interfere with the desired ED in the time domain according to Fig. \ref{Int}, the evaluation of the average numbers of time domain interfering devices for different number of EDs existing in the network are presented in Fig. \ref{Interference}. Based on these results, on average, $0.7\%$, $0.96\%$, and $0.6\%$ of the total number of EDs in the network cause interference as potential co-channel interfering devices, header interferences, and payload interferences for DR6, respectively. Similarly, for DR5, these values are $0.88\%$, $0.11\%$ and $0.75\%$.

\begin{figure}[t!]
  \centering
  \includegraphics[width=\linewidth]{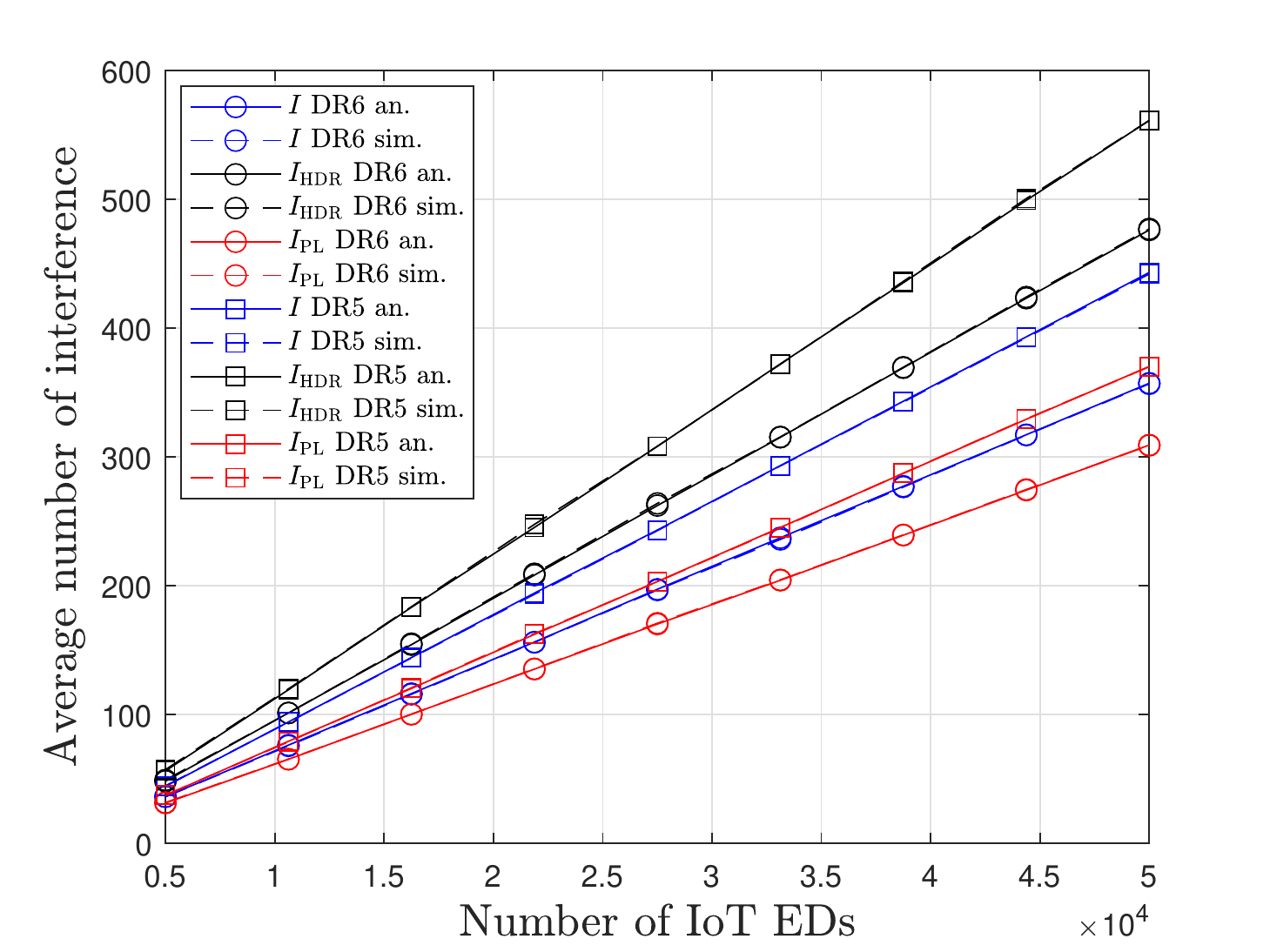}
  \caption{Analytical and simulation results for the average number of time domain interferences in unslotted ALOHA LR-FHSS scheme.}
  \label{Interference}
\end{figure}

The outage probability of the LR-FHSS scheme is indicated in Fig. \ref{capnocap}. As can be seen, the curves corresponding to no capture effect case show a degraded performance compared to the actual LR-FHSS performance in which the capture effect of the receiver is taken into account using the analysis provided in this paper. In fact, in case of no capture effect, we assume that $P_{\rm cap}(k)=1$ regardless of the number of interferences as presented in the analysis of \cite{12}. Moreover, for the satellite speed of $7.4$ km/s, one can calculate the coverage area $|\mathcal{F}|=2.4847\times 10^{7}$ $\rm{km}^2$. Hence, we can observe that for a typical outage threshold of $10^{-2}$, the LR-FHSS scheme can serve nearly $497,000$ and $1,490,000$ EDs for DR6 and DR5, respectively, in the duration of a time slot.

\begin{figure}[t!]
  \centering
  \includegraphics[width=\linewidth]{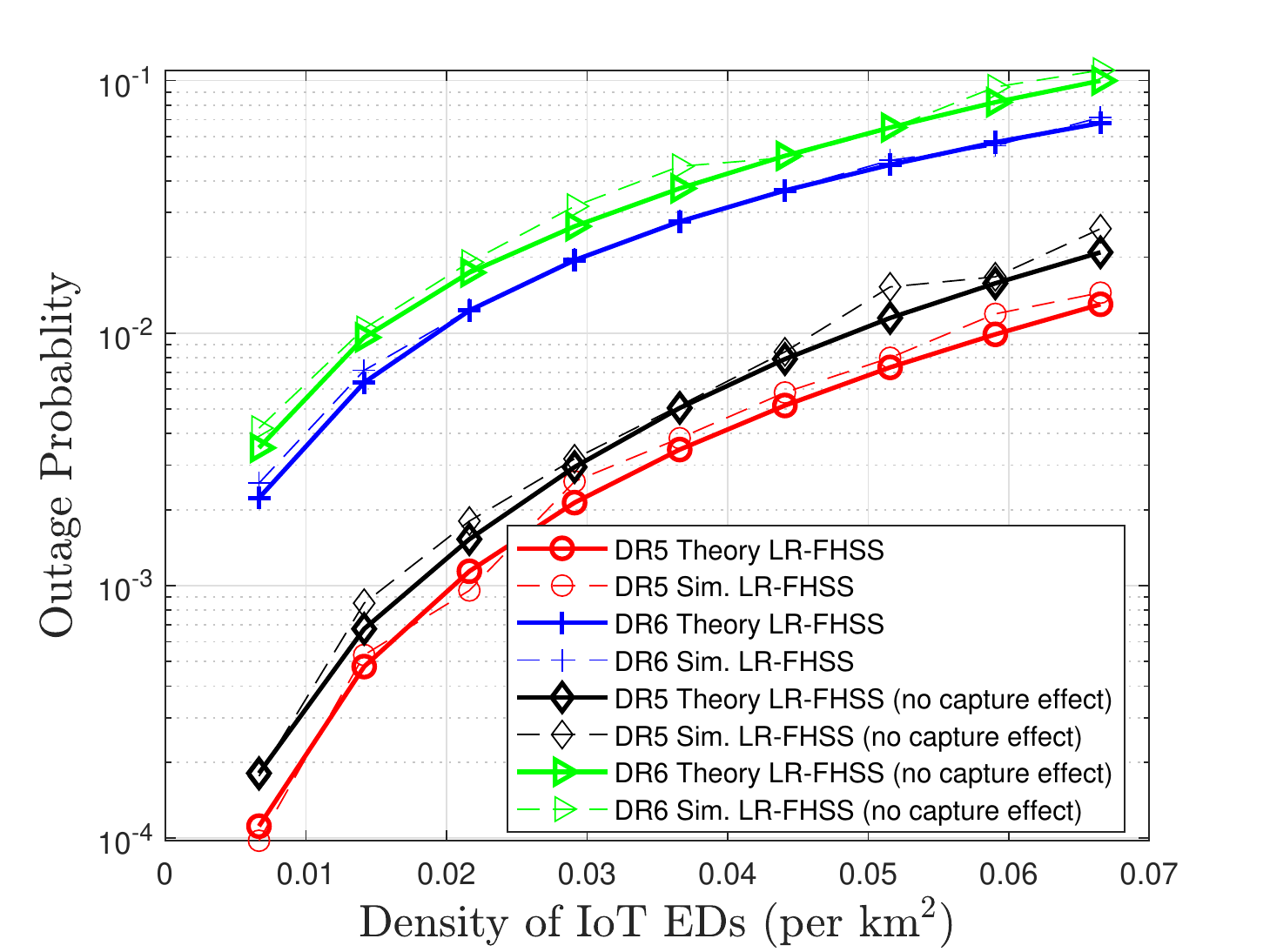}
  \caption{Outage probability of LR-FHSS in average shadowed-Rice fading environment.}
  \label{capnocap}
\end{figure}

Fig. \ref{coop} indicates the probability of successful D2D communication. It is noteworthy that we assume the probability of D2D packet exchange as $P_{\rm LoRa}=0.9$ in our simulations. As the first transmission session of the D2D-aided LR-FHSS scheme, it can be observed that for a network density greater than $0.3$ EDs per ${\rm km}^2$, we can expect an acceptable probability of $80\%$ for two EDs to form a cluster and exchange their packets using LoRa modulation. The total outage probability of D2D-aided LR-FHSS is presented in Fig. \ref{out}. It can be seen that the integration of D2D communication with the LR-FHSS scheme can significantly improve the performance in terms of network capacity in lower network densities. As an example, for a typical outage threshold of $10^{-2}$, using the same satellite speed as in the LR-FHSS case, the D2D-aided LR-FHSS framework can serve approximately $1,242,000$ and $2,236,000$ EDs for DR6 and DR5 data rates, respectively. These translate to $249.9\%$ and $150.1\%$ increase in network capacity compared to LR-FHSS for DR6 and DR5 scenarios, respectively. Obviously, this is obtained at the cost of minimum one (when D2D cannot be established) and maximum of two additional transmissions (in case of having a successful D2D data exchange) per each IoT ED in the duration of $T$. It is also worth mentioning that, considering both Fig. \ref{coop} and Fig. \ref{out}, one can realize that when the network density increases in the D2D-aided LR-FHSS scenario, and as a result, the probability of interference occurrence grows, the probability of successful D2D communication initially becomes larger too (before reaching its maximum value of $P_{\rm LoRa}$). This helps mitigate the effect of interference and achieve better performance compared to LR-FHSS.      

\begin{figure}[t!]
  \centering
  \includegraphics[width=\linewidth]{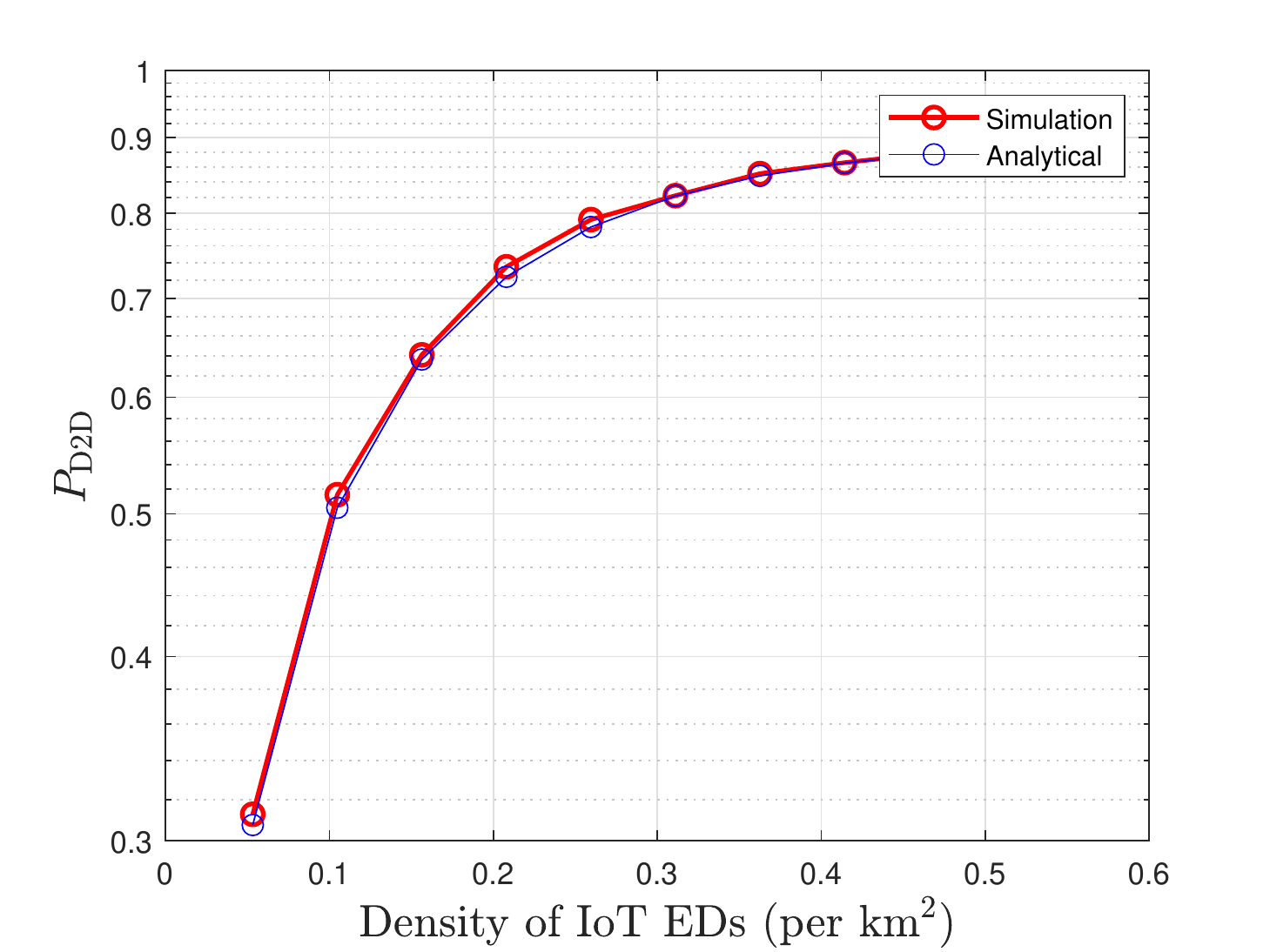}
  \caption{Probability of successful D2D communication as function of EDs density.}
  \label{coop}
\end{figure}

\begin{figure}[t!]
  \centering
  \includegraphics[width=\linewidth]{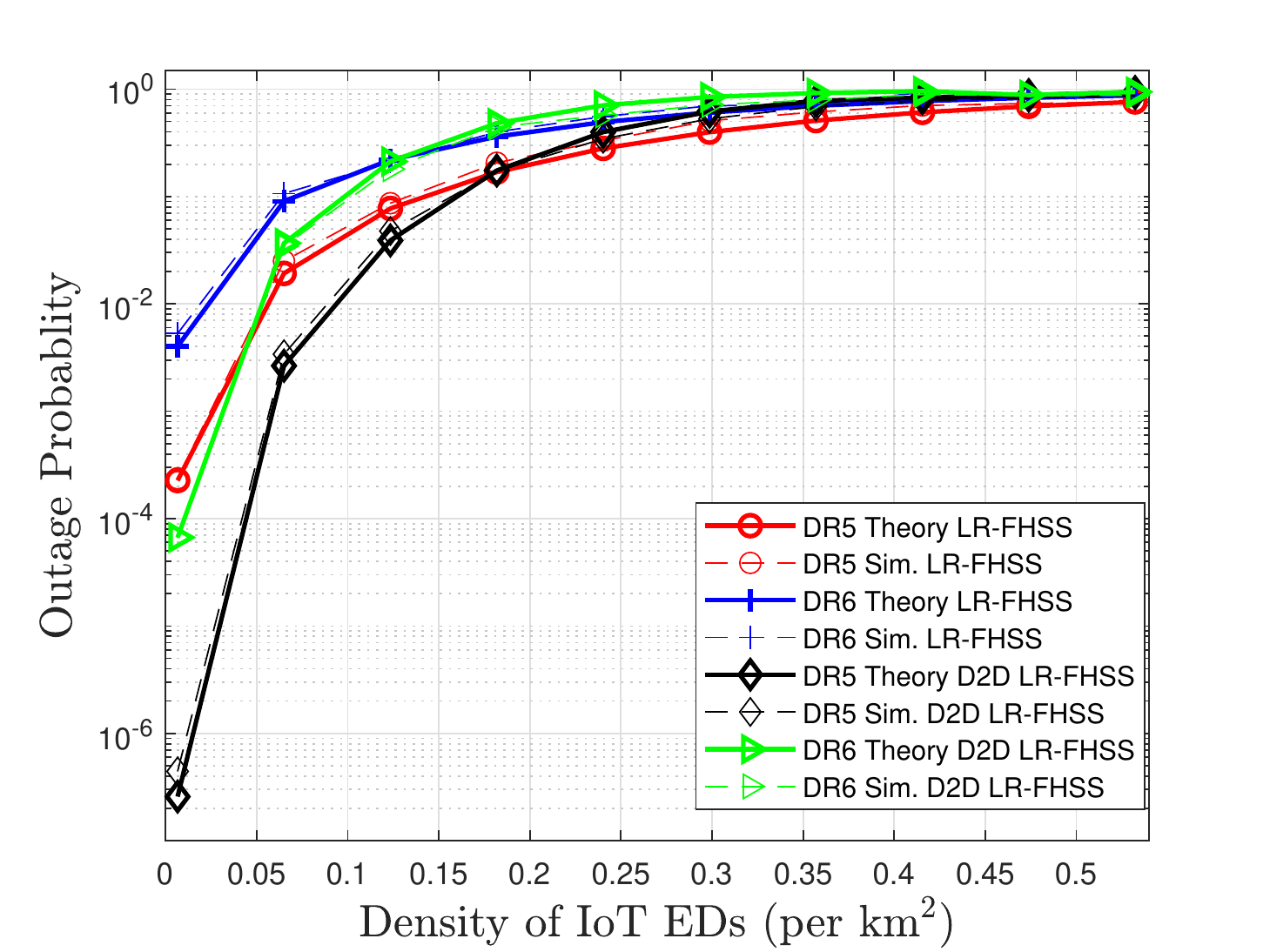}
  \caption{Outage probability of D2D-aided LR-FHSS.}
  \label{out}
\end{figure}

In the presented DtS IoT network scenario, based on the random timing of starting transmission for an ED, the distance from the ED to the satellite can be a random value between $H_{\rm s}$ and the slant range of the satellite. This is equal to starting the transmission at a random elevation angle between $90^{\circ}$ to $\alpha_{\rm SR}$ (which depends on the slant range of the satellite). As discussed before, we addressed this randomness in our simulation by averaging over $10000$ realizations of ED locations. However, as the final part of our simulation results, to see the effect of the location of an ED when it starts its transmission on the outage probability of its packet, we present the outage probability as a function of distance to the satellite in Fig. \ref{dist}. As can be seen, satellite movements can significantly affect the outage performance for a single ED transmission. Based on the results of Fig. \ref{dist}, one can understand the importance of proposing an effective time scheduling algorithm for D2D-aided LR-FHSS application in DtS IoT networks which is left as a future potential research problem.

\begin{figure}[t!]
  \centering
  \includegraphics[width=\linewidth]{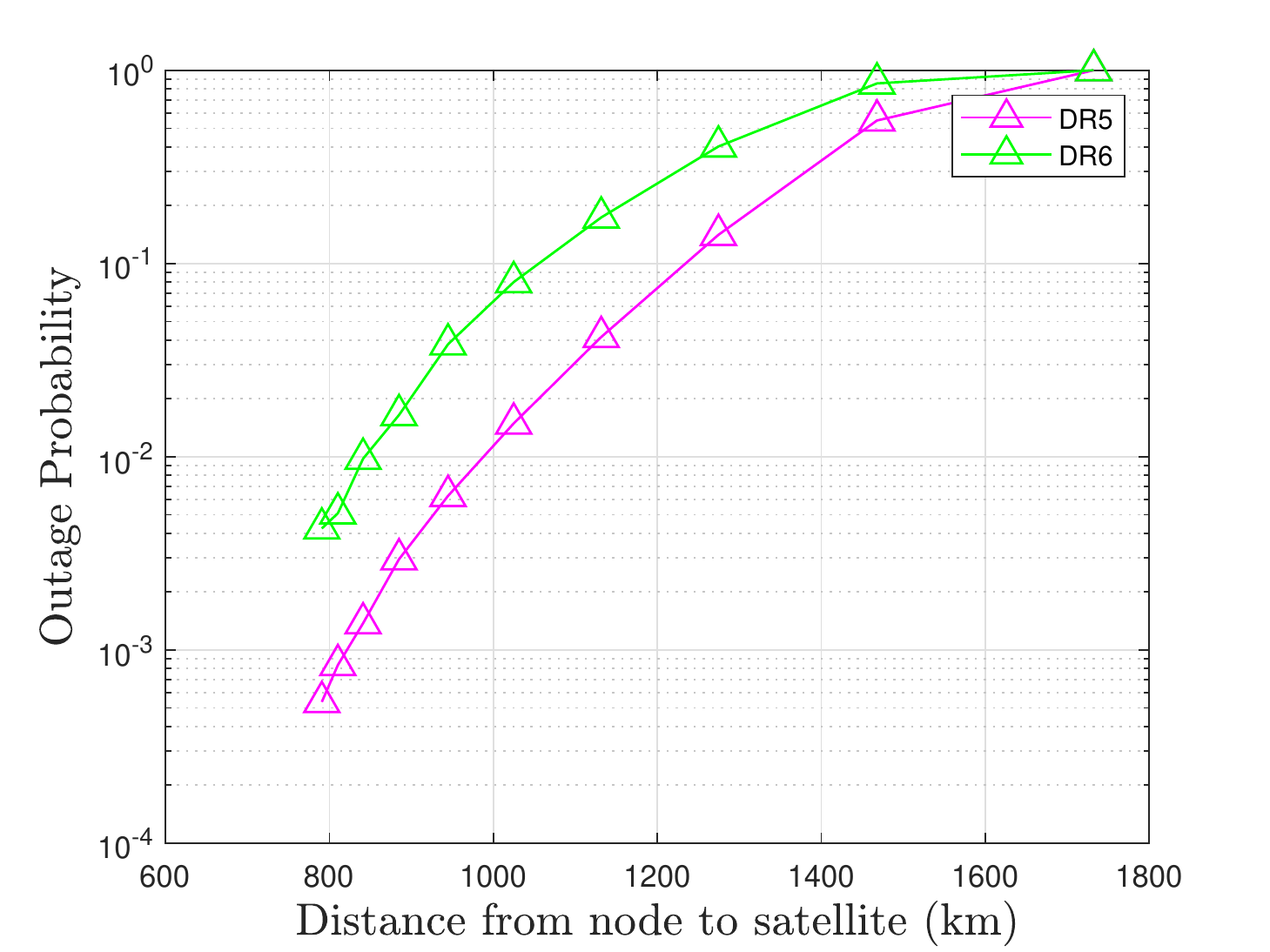}
  \caption{Outage probability of D2D-aided LR-FHSS as function of distance from node to satellite.}
  \label{dist}
\end{figure}

\section{Conclusion}
By considering cooperation between EDs in a DtS-IoT network in form of D2D communication and by exploiting a simple network coding scheme, we propose a D2D-aided LR-FHSS transmission framework with the motivation of being used in dense IoT applications. A detailed analysis is provided to obtain closed-form expressions of the outage probabilities of both LR-FHSS and D2D-aided LR-FHSS using practical channel modeling. The computer simulations validated the analytical expressions. The results show a significant boost in network performance of D2D-aided LR-FHSS in terms of the density of EDs served by a single IoT GW installed on a LEO satellite compared to LR-FHSS protocol ($249.9\%$ and $150.1\%$ increase in network capacity at a typical outage of $10^{-2}$ for DR6 and DR5, respectively). This is obtained at the cost of additional transmissions by each IoT ED.

\appendices

\section{Proof of lemma 1}
Based on the definition of probability of disconnection in (\ref{P_disc}), we have:

\begin{equation}
\label{ap1}
P_{\rm{disc}}={\rm{Pr}}\left\{\frac{P|h_0|^2 g(\alpha_0)}{\sigma^2}\leq\psi\right\}=\int_{0}^{\frac{\psi\sigma^2}{Pg_0}}{\mathfrak{f}_{R}\left(r\right)}\mathrm{d}r,
\end{equation}
where $\mathfrak{f}_R(r)$ is the PDF of the power of shadowed-Rice faded signal and can be represented as \cite{13}:

\begin{equation}
\label{ap2}
\mathfrak{f}_{R}(r)=A \exp\left(-Br\right) \times {{}_{1}F_{1}}\left[ m,1,C(1)r\right].
\end{equation}
On the other hand, for the confluent hypergeometric function can be expressed as \cite{Int_2014}:

\begin{equation}
\label{ap3}
{{}_{1}F_{1}}(a,b,x)=\sum\limits_{n=0}^{\infty} \frac{(a)_n}{n! (b)_n}x^n.
\end{equation}

Therefore, we can calculate $P_{\rm disc}$ as follows:

\begin{IEEEeqnarray}{rCl}
\label{ap4}
P_{\rm{disc}}&=&A \int_{0}^{\frac{\psi\sigma^2}{Pg_0}} {\exp\left(-Br\right) \sum\limits_{n=0}^{\infty} \frac{(m)_n}{n! (1)_n} C(n) r^n \mathrm{d}r},\nonumber \\
&=& A \sum\limits_{n=0}^{\infty} \frac{(m)_n}{n! n!} C(n) \int_{0}^{\frac{\psi\sigma^2}{Pg_0}} \exp\left(-Br\right) r^n \mathrm{d}r, \nonumber \\
&=&  A \sum\limits_{n=0}^{\infty} \frac{(m)_n}{n! n!} C(n) (\frac{1}{B})^{n+1}\int_{0}^{B\frac{\psi\sigma^2}{Pg_0}} e^{-t}  t^n \mathrm{d}t,\nonumber \\
&=& A \sum\limits_{n=0}^{\infty} \frac{(m)_n}{n! n!} C(n) (\frac{1}{B})^{n+1} \gamma\left(n+1,{B\frac{\psi\sigma^2}{Pg_0}}\right).\nonumber \\
\label{eq:dont_use_multline}
\end{IEEEeqnarray}

\section{Proof of lemma 2}

In obtaining a closed-form expression for the capture failure probability as expressed in (\ref{P_disc}), the first step is to find a CDF for the denominator $X_k=\sum_{m=1}^{k}|h_m|^2 g_m$. To obtain the CDF, we follow the same approach used in \cite{Kotz_1967,Ropokis_2009} which is based on the inverse Laplace transform of moment-generating function (MGF). Considering (\ref{ap2}), the MGF of $r$ is given as \cite{13}:

\begin{equation}
\label{ap5}
M_r(s) = \frac{(1-2b_0s)^{m-1}}{\left[1-(2b_0+\frac{\Omega}{m})s\right]^m}.
\end{equation} 
In $X_k$, we define $q_i = r_ig_i$, and considering that the MGF of $y = ax$ equals to $M_x(as)$, we can write:

\begin{equation}
\label{ap6}
M_{q_i}(s) = \frac{(1-2g_ib_0s)^{m-1}}{\left[1-(2b_0+\frac{\Omega}{m})g_is\right]^m}.
\end{equation}
Knowing that $q_i$'s are independent non identically distributed, the Laplace transform of PDF of $X_k$ can be written as follows:

\begin{equation}
\label{ap7}
L_{X_k}(s) = \prod\limits_{i=1}^{k}\frac{(1+2g_ib_0s)^{m-1}}{\left[1+(2b_0+\frac{\Omega}{m})g_is\right]^m}.
\end{equation}
Now, we need to calculate the inverse Laplace transform of $L_{X_k}(s)$ to obtain the PDF of $X_k$. First, we define $b_i=b_0g_i$ and $\Omega_i = \Omega g_i$. Then, we have:

\begin{equation}
\label{ap8}
L_{X_k}(s) = \prod\limits_{i=1}^{k}\frac{(1+2b_is)^{m-1}}{\left[1+(2b_i+\frac{\Omega_i}{m})s\right]^m}.
\end{equation}
We define a function $\eta(s)=1/(1+\alpha s)$ where $\alpha$ is an arbitrary positive parameter. On the other hand, for any $\beta>0$, we have:

\begin{equation}
\label{ap9}
1+\beta s = 1 + s\alpha \frac{\beta}{\alpha} + \frac{\beta}{\alpha} - \frac{\beta}{\alpha} = \frac{\beta}{\alpha \eta(s)} \left[1-\left(1-\frac{\alpha}{\beta}\right)\eta(s)\right].
\end{equation}
Therefore, (\ref{ap8}) can be rewritten as:

\begin{equation}
\label{ap10}
L_{X_k}(s) = D\eta^k(s)\prod\limits_{i=1}^{k}\frac{[1-\zeta_i \eta(s)]^{m-1}}{\left[1-\delta_i\eta(s)\right]^m},
\end{equation}
where

\begin{equation}
\label{ap11}
D=\alpha^k\prod\limits_{i=1}^{k}\frac{(2b_i)^{m-1}}{\left(2b_i+\frac{\Omega_i}{m}\right)^m},
\end{equation}

\begin{equation}
\label{ap12}
\zeta_i = 1-\frac{\alpha}{2b_i},
\end{equation}
and
\begin{equation}
\label{ap13}
\delta_i = 1-\frac{\alpha}{2b_i+\frac{\Omega_i}{m}}.
\end{equation}
Now, we define the function
\begin{equation}
\label{ap14}
Z(\eta) = \prod\limits_{i=1}^{k}\frac{[1-\zeta_i \eta]^{m-1}}{\left[1-\delta_i\eta\right]^m}.
\end{equation}
Taking the natural logarithm $\ln$ of both sides, we have:

\begin{equation}
\label{ap15}
\ln Z(\eta) = (m-1)\sum\limits_{i=1}^{k}\ln (1-\zeta_i \eta)-m\sum\limits_{i=1}^{k}\ln (1-\delta_i \eta).
\end{equation}
When $\eta<1/\max{\left\{\max\limits_{i}{\left\{\left|\zeta_i\right|\right\}},\max\limits_{i}{\left\{\left|\delta_i\right|\right\}}\right\}}$, i.e., the arguments of both $\ln$ functions are positive for all of $i$'s, we have the following series expansion for $Z(\eta)$:

\begin{equation}
\label{ap16}
\ln Z(\eta) = \sum\limits_{j=1}^{\infty}p_j \frac{\eta^j}{j},
\end{equation}
where $p_j=\sum_{i=1}^{k}\left[m\delta_i^j-(m-1)\zeta_i ^j\right]$. Then, (\ref{ap14}) can be rewritten as \cite[pp.~93]{Mathai_1992}:

\begin{equation}
\label{ap17}
 Z(\eta) = \sum\limits_{i=0}^{\infty} c_i \eta^i,
\end{equation}
and the coefficients of $c_i$'s are calculated recursively with $c_0=Z(0)$ as:

\begin{equation}
\label{ap18}
c_i=\frac{1}{i}\sum\limits_{l=0}^{i-1}p_{i-l}c_l.
\end{equation}
From (\ref{ap10}), (\ref{ap14}), and (\ref{ap17}), we have:

\begin{equation}
\label{ap19}
L_{X_k}(s) = D\sum\limits_{i=0}^{\infty}c_i\eta^{k+i}(s).
\end{equation}

On the other hand, for an integer $a$, we have the following inverse Laplace transform:

\begin{equation}
\label{ap20}
\mathcal{L}^{-1}\{\eta^a (s)\}=\frac{x^{(a-1)}\exp(-x/\alpha)}{\alpha^a (a-1)!}.
\end{equation}
Consequently, the PDF of $X_k$ can be obtained as:

\begin{equation}
\label{ap21}
f_{X_k}(x) = D\sum\limits_{i=0}^{\infty}c_i\frac{x^{k+i-1}\exp(-x/\alpha)}{\alpha^{k+i}(k+i-1)!}.
\end{equation}
To obtain the CDF, we use the formula $F_{X_k}(x)=\int_{0}^{x}f_{X_k}(\tau)\mathrm{d}\tau$. Therefore, we have:

\begin{equation}
\label{ap22}
F_{X_k}(x) = D\int_{0}^{x}\sum\limits_{i=0}^{\infty}c_i\frac{\tau^{k+i-1}\exp(-\tau/\alpha)}{\alpha^{k+i}(k+i-1)!}\mathrm{d}\tau.
\end{equation}
To be able to interchange integration and summation, uniform convergence of the series in (\ref{ap22}) is required. This result is established in \cite{Ropokis_2009}, where also the range of admissible values for parameter $\alpha$ is given as $0<\alpha<4\times \min\limits_{i}\{b_i\}$. Therefore, we have:

\begin{equation}
\label{ap23}
F_{X_k}(x) = D\sum\limits_{i=0}^{\infty}\frac{c_i}{\alpha^{k+i}(k+i-1)!}\int_{0}^{x}\tau^{k+i-1}\exp(-\tau/\alpha)\mathrm{d}\tau.
\end{equation}
The integral part can be calculated as:

\begin{equation}
\label{ap24}
\int_{0}^{x}\tau^{k+i-1}\exp(-\tau/\alpha)\mathrm{d}\tau=\alpha^{k+i}\gamma(k+i,x/\alpha).
\end{equation}
Since $k+i$ is a positive integer, we have:

\begin{equation}
\label{ap25}
\gamma(k+i,x/\alpha)=(k+i-1)!\left[1-e^{-x/\alpha}\left(\sum\limits_{u=0}^{k+i-1}\frac{(x/\alpha)^u}{u!}\right)\right].
\end{equation}
The CDF expression can be presented as:

\begin{equation}
\label{ap26}
F_{X_k}(x) = D\sum\limits_{i=0}^{\infty}c_i\left[1-e^{-x/\alpha}\left(\sum\limits_{u=0}^{k+i-1}\frac{(x/\alpha)^u}{u!}\right)\right].
\end{equation}
Using this result, we have:

\begin{IEEEeqnarray}{rCl}
\label{ap27}
P_{\rm cap}(k)&=&1-\mathbb{E}_{\left|h_0\right|^2}\left[\Pr{\left\{X_k<\frac{\left|h_0\right|^2g_0}{\delta}\right\}}\right] \nonumber \\
&=& 1- \int_{0}^{\infty} A e^{-Bz}\times {{}_{1}F_{1}}(m,1,C(1)z)F_{X_k}\left(\frac{zg_0}{\delta}\right)\mathrm{d}z\nonumber \\
&=& 1- \int_{0}^{\infty} A e^{-Bz}\times {{}_{1}F_{1}}(m,1,C(1)z) \nonumber\\
&&\times D\sum\limits_{i=0}^{\infty}c_i \left[1-e^{-\frac{zg_0}{\alpha\delta}}\sum\limits_{u=0}^{k+i-1}\frac{(zg_0)^u}{u!(\delta\alpha)^u}\right]\mathrm{d}z\nonumber \\
&=& 1 - AD\sum\limits_{i=0}^{\infty}c_i \underbrace{\int_{0}^{\infty}e^{-Bz}\times {{}_{1}F_{1}}(m,1,C(1)z)\mathrm{d}z}_{\rm{(I)}} \nonumber \\
&& + AD\sum\limits_{i=0}^{\infty}c_i \sum\limits_{u=0}^{k+i-1}\frac{1}{u!}\nonumber\\
&&\underbrace{\int_{0}^{\infty}e^{-z(B+\frac{g_0}{\alpha\delta})}\times(\frac{zg_0}{\alpha\delta})^u \times {{}_{1}F_{1}}(m,1,C(1)z) \mathrm{d}z}_{\rm{(II)}}.\nonumber \\
\label{eq:dont_use_multline}
\end{IEEEeqnarray}

The final step towards obtaining a closed-form expression for $P_{\rm{cap}}$ is to calculate the integrals (I) and (II). For (I), we have:

\begin{IEEEeqnarray}{rCl}
\label{ap28}
{\rm{(I)}}&=&\int_{0}^{\infty}e^{-Bz}\times {{}_{1}F_{1}}(m,1,C(1)z)\mathrm{d}z \nonumber \\
&=& \int_{0}^{\infty}e^{-Bz}\sum\limits_{n=0}^{\infty}\frac{(m)_n}{n!n!}C(n) z^n \mathrm{d}z \nonumber\\
&=& \sum\limits_{n=0}^{\infty}\frac{(m)_n}{n!n!}C(n) \int_{0}^{\infty}e^{-Bz} z^n \mathrm{d}z\nonumber \\
& =& \sum\limits_{n=0}^{\infty}\frac{(m)_n}{n!n!}C(n) \times (\frac{1}{B})^{n+1} \Gamma (n+1). 
\label{eq:dont_use_multline}
\end{IEEEeqnarray}

Defining $B'=B+\frac{g_0}{\alpha\delta}$, for the (II), we can write:

\begin{equation}
\label{ap29}
{\rm{(II)}}=(\frac{g_0}{\alpha\delta})^u \int_{0}^{\infty}e^{-B' z}\times z^u \times {{}_{1}F_{1}}(m,1,C(1)z) \mathrm{d}z.
\end{equation}

From \cite{Int_2014}, we know the following:

\begin{equation}
\label{ap30}
\int_{0}^{\infty} e^{-wt} t^{e-1} {{}_{1}F_{1}}(a;v;xt) \mathrm{d}t = \Gamma(e)w^{-e}{{}_{2}F_{1}}(a,e;v;xw^{-1}),
\end{equation}
for $|w|>|x|$. It should be noted that for the values provided in Table \ref{ShEnv}, the condition $|B'|>|C(1)|$ will be satisfied in integral part of (\ref{ap29}). Therefore, we can rewrite (\ref{ap29}) as:

\begin{equation}
\label{ap31}
{\rm{(II)}}=(\frac{g_0}{\alpha\delta})^u \times \Gamma(u+1)\times {B'} ^{-(u+1)}\times {{}_{2}F_{1}} (m,u+1,1,\frac{C(1)}{B'}).
\end{equation}
Finally, by substituting (\ref{ap28}) and (\ref{ap31}) into (\ref{ap27}), we have:

\begin{IEEEeqnarray}{rCl}
\label{ap32}
P_{\rm cap}(k)&=& 1 - AD\sum\limits_{i=0}^{\infty}c_i  \sum\limits_{n=0}^{\infty}\frac{(m)_n}{n!n!}C(n) \times (\frac{1}{B})^{n+1} \Gamma (n+1)\nonumber \\
&& + AD\sum\limits_{i=0}^{\infty}c_i \sum\limits_{u=0}^{k+i-1}\frac{1}{u!} (\frac{g_0}{\alpha\delta})^u \times \Gamma(u+1)\times {B'} ^{-(u+1)}\nonumber\\
&&\times {{}_{2}F_{1}} (m,u+1,1,\frac{C(1)}{B'}).
\label{eq:dont_use_multline}
\end{IEEEeqnarray}

\bibliographystyle{IEEEtran}
\bibliography{IoT_Transmission_Techs_PhD_V11}

\end{document}